\documentclass[lettersize,journal]{IEEEtran}
\usepackage{amsmath,amsfonts}
\usepackage[ruled]{algorithm2e}
\usepackage{array}
\usepackage[caption=false,font=normalsize,labelfont=sf,textfont=sf]{subfig}
\usepackage{textcomp}
\usepackage{stfloats}
\usepackage{url}
\usepackage{amssymb}
\usepackage{verbatim}
\usepackage{color}
\usepackage{soul}
\usepackage{graphicx}
\usepackage{cuted}
\usepackage{bm} 
\usepackage{mathrsfs}
\usepackage{multirow}
\usepackage{color}
\usepackage{cite}

\usepackage{amsthm}

\newtheorem{observation}{Observation}
\newtheorem{corollary}{Corollary}
\newtheorem{remark}{Remark}
\newtheorem{proposition}{Proposition}
\newtheorem{property}{Property}

\begin{document}
\title{Ambiguity Function Analysis of AFDM Signals for Integrated Sensing and Communications}
\author{Haoran Yin, Yanqun Tang, Yuanhan Ni, \textit{Member, IEEE}, Zulin Wang, \textit{Member, IEEE},\\  Gaojie Chen, \textit{Senior Member, IEEE}, Jun Xiong, \textit{Member, IEEE}, Kai Yang, \textit{Member, IEEE}, \\Marios Kountouris, \textit{Fellow, IEEE}, Yong Liang Guan, \textit{Senior Member, IEEE}, and Yong Zeng, \textit{Fellow, IEEE}
	\thanks{
		This work was supported in part by the Shenzhen Science and Technology Major Project under Grant KJZD20240903102000001 and in part by the Science and Technology Planning Project of Key Laboratory of Advanced IntelliSense Technology, Guangdong Science and Technology Department under Grant 2023B1212060024. (\textit{Corresponding author: Yanqun Tang.})
		
		 Haoran Yin and Yanqun Tang are with the School of Electronics and Communication Engineering, Sun Yat-sen University, China (e-mail: yinhr6@mail2.sysu.edu.cn,  tangyq8@mail.sysu.edu.cn);
		 
		 Yuanhan Ni and Zulin Wang are with the School of
		 Electronic and Information Engineering, Beihang University, China (e-mail: yuanhanni@buaa.edu.cn,  wzulin@buaa.edu.cn);
		 
		 Gaojie Chen is with the School of Flexible Electronics (SoFE), Sun Yat-sen University, China (e-mail: gaojie.chen@ieee.org);
		 
		 Jun Xiong is with the College of Electronic Science and Technology,
		  National University of Defense Technology, China (e-mail:xj8765@nudt.edu.cn);
		  
		 Kai Yang is with the School of Information and Electronics, Beijing Institute of Technology, China (e-mail: yangkai@ieee.org);
		 
		 Marios Kountouris is with the Communication Systems Department, EURECOM, France, and with the Andalusian Research Institute in Data Science and Computational Intelligence (DaSCI), Department of Computer Science and Artificial Intelligence, University of Granada, Spain (e-mail: marios.kountouris@eurecom.fr);
		 
		 Yong Liang Guan is with the School of Electrical and Electronic Engineering, Nanyang Technological University, Singapore (e-mail: eylguan@ntu.edu.sg);
		 
		 Yong Zeng is with the National
		 Mobile Communications Research Laboratory and the Frontiers Science Center for Mobile Information Communication and Security, Southeast University, and also with the Purple Mountain Laboratories, China (e-mail: yong$\_$zeng@seu.edu.cn).			 
		}  
}
\markboth{}%
{Shell \MakeLowercase{\textit{et al.}}: A Sample Article Using IEEEtran.cls for IEEE Journals}
\maketitle

\begin{abstract}
	\textbf{Affine frequency division multiplexing (AFDM) is a promising chirp-based waveform with high flexibility and resilience, making it well-suited for next-generation wireless networks, particularly in high-mobility scenarios. In this paper, we investigate the ambiguity functions (AFs) of AFDM signals, which fundamentally characterize their range and velocity estimation capabilities in both monostatic and bistatic settings. Specifically, we first derive the auto-ambiguity function (AAF) of an AFDM chirp subcarrier, revealing its “spike-like” local property and “periodic-like” global property along the rotated delay and Doppler dimensions. This structure naturally forms a parallelogram for each localized pulse of the AAF of the AFDM chirp subcarrier, enabling unambiguous target sensing. Then, we study the cross-ambiguity function (CAF) between two different AFDM chirp subcarriers, which exhibits the same local and global properties as the AAF but with an additional shift along the Doppler dimension. We then extend our analysis to the AF of various typical AFDM frames, considering both deterministic pilot and random data symbols. In particular, we demonstrate that inserting guard symbols in AFDM facilitates interference-free sensing. Simulation results validate our theoretical findings, highlighting AFDM's strong potential for ISAC applications.} 
\end{abstract}


\section{Introduction}
Sixth-generation (6G) wireless networks are set to transform emerging applications, including vehicle-to-everything (V2X), unmanned aerial vehicles (UAV), and satellite networks \cite{bb25.2.4.1}. However, the conventional orthogonal frequency division multiplexing (OFDM) waveform is highly sensitive to Doppler shifts in these high-mobility scenarios \cite{bb22.10.24.2}, highlighting the urgent need for novel communication waveforms. Meanwhile, integrated sensing and communications (ISAC) has been recognized by the International Telecommunication Union (ITU) as one of the six pivotal usage scenarios in 6G, making the sensing capabilities of communication waveforms a crucial factor in evaluating their suitability and effectiveness \cite{bb25.01.08.102,bb25.01.08.1}. Consequently, the demand for advanced communication waveforms to enable effective ISAC in doubly selective channels (DSCs) has become increasingly urgent.

In this context, several new waveforms have emerged in recent years. The delay-Doppler (DD) waveforms, including orthogonal time frequency space (OTFS) \cite{bb2,bb24.08.21.2,bb25.01.21.1}, orthogonal delay Doppler division multiplexing (ODDM) \cite{bb24.03.15.2}, and delay-Doppler alignment modulation (DDAM) \cite{ bb25.02.11.2}, attract widespread interest.  OTFS leverages the discrete Zak transform to multiplex information symbols in the DD domain, where the underlying “pulsone" subcarrier structure in OTFS effectively captures the quasi-static characteristics of the DSC \cite{bb25.01.21.2, bb24.03.15.311,  bb24.03.15.3}. Moreover, DDAM is particularly appealing for millimeter wave (mmWave) extremely large-scale MIMO (XL-MIMO) systems that have large spatial resolution and multi-path sparsity \cite{bb25.02.11.1}. Another promising alternative is orthogonal chirp division multiplexing (OCDM), a chirp-based waveform with a fixed chirp slope determined by the discrete Fresnel transform (DFnT) \cite{bb24.08.27.1}. Although OCDM offers certain improvements over OFDM, it fails to fully exploit the inherent diversity of the DSC. Recently, a novel chirp-based multicarrier waveform,  named affine frequency division multiplexing (AFDM), has been proposed for communications in DSCs, garnering significant attention \cite{bb24.08.27.2}. The core principle of AFDM involves modulating information symbols onto a set of orthogonal chirps with a tunable chirp slope, utilizing the discrete affine Fourier transform (DAFT), which is a generalization of both the DFnT and the discrete Fourier transform (DFT) \cite{bb8}. By adjusting its two fundamental chirp parameters to match the DD profile of the DSC, AFDM can effectively separate all propagation paths with distinct delay or Doppler shifts in the underlying DAFT domain. This not only ensures precise channel estimation and optimal diversity gain in DSCs but also provides it with high flexibility. These compelling advantages, combined with its inherent chirp-based nature, position AFDM as a strong candidate for 6G waveforms to enable ISAC in high-mobility scenarios\cite{bb25.01.08.2, bb24.9.08.100}.

Extensive research has been conducted to unlock the full potential of AFDM. An embedded pilot-aided (EPA) channel estimation scheme was introduced in \cite{bb24.08.27.2}, demonstrating that AFDM requires lower channel estimation overhead compared to OTFS. In \cite{23.10.18.1}, an EPA diagonal reconstruction (EPA-DR) channel estimation scheme was proposed, which reconstructs AFDM's quasi-diagonal effective channel matrix directly from the received pilot symbols by exploiting its diagonal reconstructability. Besides, the authors of \cite{bb25.01.25.1} and \cite{bb25.01.25.2} proposed eliminating guard symbols and performing joint channel estimation and signal detection to further minimize the channel estimation overhead. Furthermore, an instructive pulse shaping guideline was introduced in \cite{bb24.9.15.1}, which enhances the pilot-data separability of the AFDM receiver and enables low-overhead, high-accuracy channel estimation. To further reduce the modulation complexity, the authors of \cite{bb24.9.08.3} and \cite{bb24.9.08.2} proposed optimizing chirp parameters without compromising the bit error ratio (BER) performance. In \cite{bb23.3.6.1}, a scalable space-time coding scheme, called cyclic delay-Doppler shift (CDDS), was developed to achieve optimal transmit diversity gain in multiple-input multiple-output AFDM (MIMO-AFDM) systems. Moreover, research on AFDM in areas such as signal detection \cite{bb25.01.09.4, bb25.01.25.3}, multiple access \cite{bb25.01.08.2,bb24.9.08.1,bb24.9.08.10}, and index modulation \cite{bb24.03.15.8,bb24.9.23.2, bb24.03.15.10, bb24.03.15.1230} highlights its strong practicality and versatility.

While the aforementioned studies focused on improving AFDM's communication performance, many others have explored its applicability in ISAC \cite{b9, bb24.9.08.103}. For instance, the authors of \cite{bb2022.11.11.2} proposed two parameter estimation techniques and introduced two novel metrics for processing echo signals in both the time domain and DAFT domains, effectively decoupling delay and Doppler shifts along the fast-time axis. Meanwhile, a compressed sensing algorithm with recovery guarantees was proposed in \cite{bb25.1.21.10}, demonstrating the superiority of AFDM in terms of minimizing the sampling rate requirements for sub-Nyquist radar. Moreover, the authors of \cite{bb2025.6.26.1} proposed a generalized ISAC solution for AFDM systems, enabling joint channel and data estimation. In contrast to the discrete-time signal assumption adopted in the aforementioned studies, the authors of \cite{bb24.03.15.5} introduced an AFDM-based ISAC system with simple self-interference cancellation in a continuous-time setting, showing that using just a single DAFT-domain pilot symbol achieves nearly the same resolution performance as an entire AFDM frame. Furthermore, the authors of \cite{bb24.9.08.104} proposed a novel angle-delay-Doppler estimation scheme for AFDM-ISAC systems operating in mixed near-field and far-field scenarios, enabling off-grid parameter estimation.

Recently, the analysis and optimization of communication waveforms' ambiguity functions (AFs) under random data signaling have emerged as a key research focus in communication-centric ISAC waveform design. In particular, the mainlobe widths of the AF along the delay and Doppler dimensions determine the range and velocity resolutions, respectively, while the sidelobe level indicates the extent of interference from adjacent targets, factors that are crucial for high-resolution multi-target sensing in 6G. In \cite{bb25.01.09.1}, the statistical properties of the AF for single-carrier modulation (SCM) signals were examined, leading to the development of a randomness-aware filter technique that minimizes the average sidelobe level. Moreover, \cite{bb25.01.09.2} demonstrated that for quadrature amplitude modulation (QAM) and phase shift keying (PSK) constellations, OFDM is the globally optimal modulation scheme for cyclic prefix (CP)-based signals with Nyquist sampling, achieving the lowest average ranging sidelobe level under random signaling. Building on this, \cite{bb25.01.09.3} provided a systematic analysis of the ranging performance of continuous-time ISAC signals with random data and introduced a novel Nyquist filter design to further improve sensing capabilities. Besides, the studies in \cite{bb25.01.21.2, bb24.03.15.3,bb24.03.15.33} revealed that a single pilot subcarrier in OTFS can generate an AF with rapidly decaying sidelobes. In \cite{bb24.9.08.4}, the zero-Doppler cut of the AF for AFDM was preliminarily analyzed, leading to a parameter selection criterion for near-optimal AF design. However, the analysis assumed a discrete-time AFDM signal with all chirp subcarriers carrying pilot symbols and on-grid delay shifts for all targets, assumptions that limit its ability to fully characterize the sensing performance of continuous-time AFDM signals in practical scenarios. Against this backdrop, a critical yet unresolved fundamental question naturally arises: \emph{What about the AFs of general continuous-time AFDM signals?} 

Therefore, this paper aims to answer this question by providing a comprehensive analysis of the two-dimensional AFs of continuous-time AFDM signals - a fundamental measure of their sensing capabilities for range and velocity estimation using matched filtering (MF). In particular, we examine the auto-ambiguity function (AAF) and cross-ambiguity function (CAF) of various AFDM signals, corresponding to monostatic and bistatic sensing configurations, respectively. The main contributions of this work can be summarized as follows:
\begin{itemize} 
	\item[$\bullet$]
	We begin by deriving the AAF of AFDM chirp subcarriers in detail, revealing their “spike-like” local property and “periodic-like” global property along the rotated delay and Doppler dimensions, along with their underlying rationale. Based on that, we define an unambiguity parallelogram for the AAF of AFDM chirp subcarriers and demonstrate how to achieve unambiguous sensing in a given sensing channel by flexibly adjusting its shape through proper AFDM parameter selection. 
\end{itemize}
\begin{itemize} 
	\item[$\bullet$]
	We analyze the CAF between different AFDM chirp subcarriers, revealing that they share the same local and global properties as the AAF of AFDM chirp subcarriers with an additional shift along the Doppler dimension. Specifically, through a rigorous derivation of the CAF’s closed-form expression, we demonstrate that this extra Doppler shift precisely corresponds to the frequency difference between the two chirp subcarriers.
\end{itemize}
\begin{itemize} 
	\item[$\bullet$]
	We further investigate the AFs of various typical AFDM frames with different pilot-data structures. Specifically, we statistically analyze the AAF of AFDM frames containing only random data, demonstrating their thumbtack-like property through simulations. Additionally, we show that inserting guard symbols between the pilot and data symbols in AFDM enables interference-free sensing. This can be achieved by appropriately tuning the AFDM parameters and adjusting the number of guard symbols based on the sensing channel characteristics.
\end{itemize}

The remainder of this paper is organized as follows. Sec. \ref{Sec2} introduces the AFDM-ISAC system model, which lays the foundations for the analysis of the AAF of the AFDM chirp subcarrier in Sec. \ref{sec3}. Sec. \ref{sec4} extends the analysis to the CAF between two different AFDM chirp subcarriers, followed by the discussion of the AFs of various typical AFDM frames in Sec. \ref{sec5} and simulation results in Sec. \ref{sec6}. Finally, Sec. \ref{sec7} concludes this paper.

\textit{Notations:} Symbol $\mathbf{a}[i]$ represents the $i$-th element of vector $\mathbf{a}$;  $\mathbb{Z}$ denotes the set of integers; $\varnothing$ represents the empty set; $a \sim \mathcal{C} \mathcal{N}\left(\mathbf{0}, N_{0}\right)$ indicates that random variable $a$ follows the complex Gaussian distribution with zero mean and covariance $N_{0}$; $\delta\left(\cdot\right)$ denotes the Dirac delta function; $(\cdot)^{*}$ and $(\cdot)^{T}$ denote the conjugate and transpose, respectively; $\lvert  \cdot  \rvert$ represents the absolute value of a complex scalar; $\mathbb{E}(\cdot)$ denotes the expectation; $\text{Real}(\cdot)$ extracts the real part of a variable; and $\min\{\cdot\}$ selects the minimum value from the input; $j$ represents $\sqrt{-1}$.

\section{AFDM-ISAC System Model}
\label{Sec2}
In this section, we introduce the AFDM-ISAC system in both monostatic and bistatic settings, which is illustrated in Fig. \ref{fig2-1}. 
\begin{figure}[tbp]
	\centering
	\includegraphics[width=0.460\textwidth,height=0.27\textwidth]{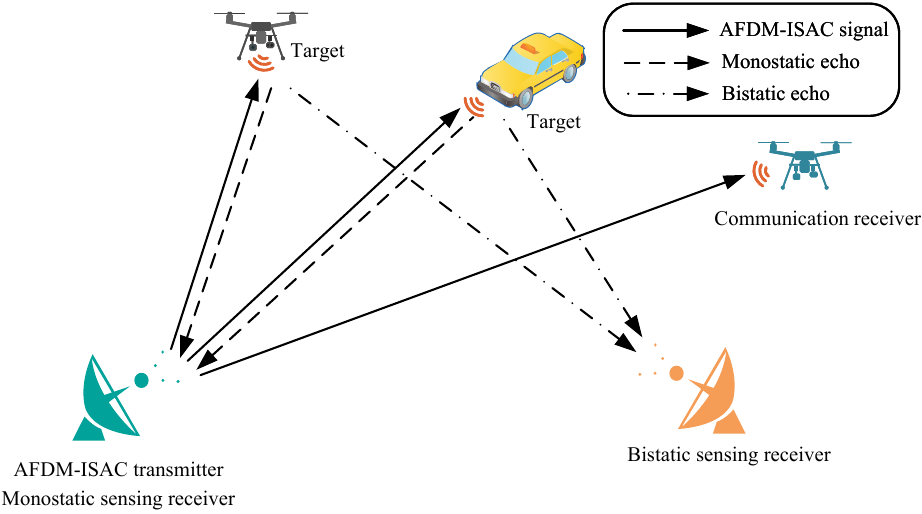}
	\caption{AFDM-ISAC system model for both monostatic and bistatic sensing scenarios.}
	\label{fig2-1}
\end{figure}
\begin{figure}[tbp]
	\centering
	\includegraphics[width=0.479\textwidth,height=0.5\textwidth]{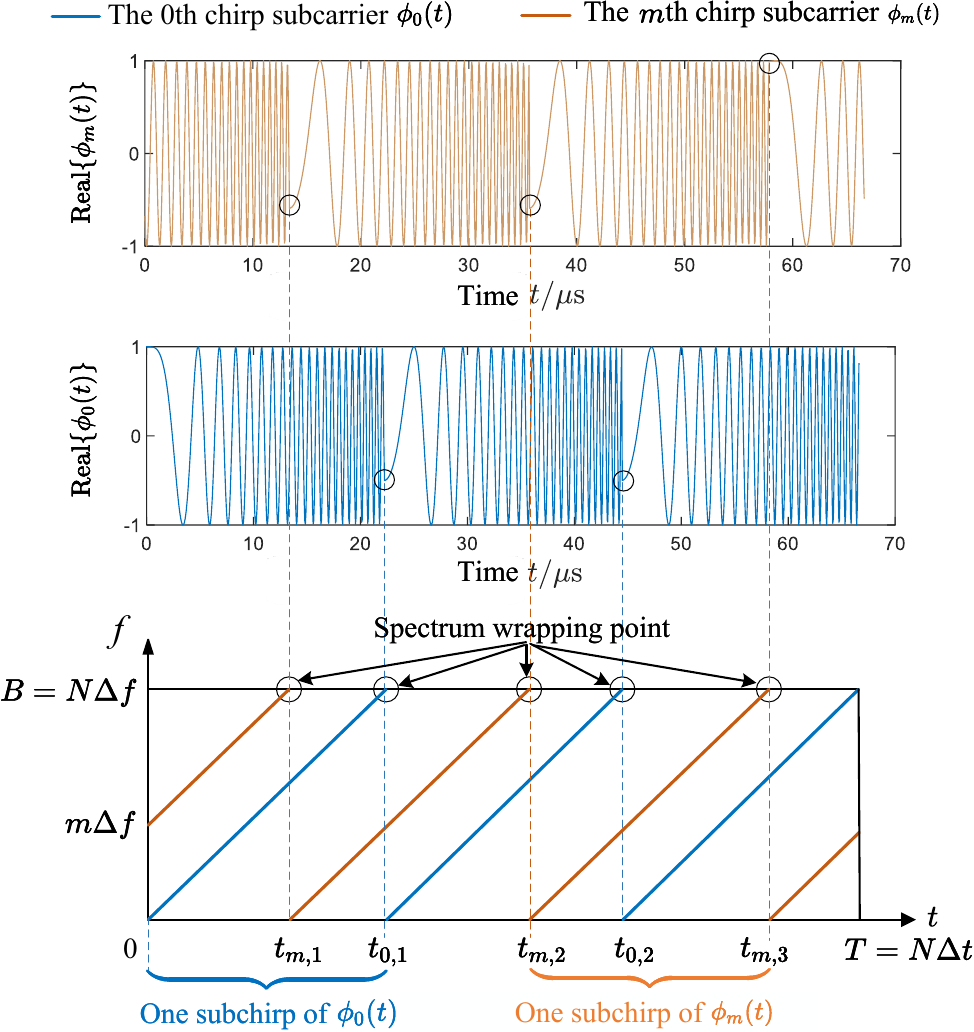}
	\caption{An example of the time and time-frequency representations of the $0$th and $m$th chirp subcarriers, with parameters: $T \approx 66.67 \mu$s,  $N=128$,  $c_{1} =\frac{3}{2N}$, and  $c_{2} =\sqrt{2}$.}
	\label{fig.2-2}
\end{figure} 
\subsection{AFDM-ISAC Transmitter}
Let $N$ denote the number of chirp subcarriers and $\mathbf{x}= \left[ x[0], x[1], \dots, x[N-1]\right]^{T}$ represent the transmitted DAFT-domain vector with its entries consisting of pilot symbol, guard symbols or data symbols. Then, $N$-point inverse DAFT (IDAFT) is firstly implemented on $\mathbf{x}$ to modulate the $N$ DAFT-domain symbols onto the $N$ chirp subcarriers, respectively, which generates the discrete-time-domain AFDM-ISAC signal as
\begin{equation}
	s[n]= \frac{1}{\sqrt{N}}\sum_{m=0}^{N-1} x[m] \phi_{m}[n],  \quad n=0,1,\dots, N-1,
	\label{eq2-1}
\end{equation}
where $n$ and $m$ denote the discrete-time domain and DAFT domain indices, respectively, and $\phi_{m}[n]$ represents the $m$th subcarrier with an expression given by 
\begin{equation} 
	\phi_{m}[n]= e^{j 2 \pi\left(c_{2} m^{2}+c_{1} n^{2}+  \frac{mn }{N} \right)}, \quad m=0,1,\dots, N-1.
	\label{eq2-2}
\end{equation}
$c_{1}$ and $c_{2}$ are two key parameters that determine the chirp slope and the initial phase of all AFDM chirp subcarriers, respectively. As we shall see later, these two parameters, especially $c_{1}$, significantly influence the AFs of AFDM signals. Moreover, a chirp-periodic prefix (CPP) given by
\begin{equation} 
	s[n] = s[n+N]e^{-j 2 \pi c_{1}\left(N^{2}+2Nn \right)}, \ n=-L_{\text{CPP}}, \dots, -1
	\label{eq24.01.09.1}
\end{equation}
is appended to $s[n]$ to cope with the multipath effect of the communication channel, where $L_{\text{CPP}}$ is the length of the CPP and should be set larger than or equal to the maximum delay spread of the communication channel. Note that when $2N|c_{1}| \triangleq C$ is an integer and $N$ is even, as assumed throughout this paper, equation (\ref{eq24.01.09.1}) simplifies to $s[n] = s[n+N],\ n=-L_{\text{CPP}}, \dots, -1$, i.e., the CPP reduces to a conventional CP as in OFDM. The corresponding continuous-time representation of $\phi_{m}[n]$ can be obtained in a piecewise manner as \cite{bb24.03.15.5}
\begin{equation} 
\phi_{m}(t) = e^{j 2 \pi\left(c_{2}m^{2}+\tilde{c}_{1}t^{2}+\frac{m}{T}t-\frac{q}{\Delta t}t\right)}, \ \ \
t_{m,q} \le t < t_{m,q+1},
\label{eq2-3}
\end{equation}
where $t$ and $\Delta t$ denote the continuous-time domain and the Nyquist sampling interval, respectively, $\tilde{c}_{1} \triangleq \frac{c_{1}}{\Delta t^{2}}$, $t_{m,q}$ is the $q$th ($q = 0,1,\dots,C$) \emph{spectrum wrapping point} of the $m$th chirp subcarrier with a definition given by
\begin{align}
	& \ t_{m,q} \triangleq \left\{\begin{array}{ll}
		0,& q=0 \\
		\frac{T}{C}q-\frac{m}{C}\Delta t,
		& q=1,2, \dots,C \\ 
		T, &q = C+1,
	\end{array}\right. 
\label{eq24.01.08.1}
\end{align}
where $T=N\Delta t$ is the duration of $\phi_{m}(t)$. This paper adopts the widely used rectangular pulse shaping, wherein all chirp subcarriers have finite support over the interval $[0,T]$ \cite{bb24.9.15.1}. One can easily verify that the equation of $\phi_{m}(t)|_{t=n\Delta t}=\phi_{m}[n], \ n=0,1, \dots, N-1$ holds for all $m = 0,1, \dots, N-1$, i.e., $\phi_{m}[n]$ can be acquired by sampling $\phi_{m}(t)$ with a sampling interval of $\Delta t$. For clarity, we illustrate the time-domain and time-frequency-domain representations of two AFDM chirp subcarriers in Fig. \ref{fig.2-2}. Building on this, the CPP-free continuous-time transmitted AFDM-ISAC signal can be expressed as\footnote{ The transmitted AFDM-ISAC signal $s(t)$ in (\ref{eq24.01.20.2}) serves as an “ideal" continuous-time counterpart of $s[n]$ in (\ref{eq2-1}), enabling insightful close-form derivations of the AFs for AFDM signals. In practice, it can be well approximated by interpolating $s[n]$ and passing the result through a digital-to-analog filter \cite{bb24.9.08.3}.}
\begin{align}
	s(t)= \sum_{m=0}^{N-1} x[m] \phi_{m}(t),  \quad 0 \le t < T. \label{eq24.01.20.2}
\end{align}

\begin{observation}[Characteristics of the AFDM chirp subcarrier] \textup{The spectrum wrapping phenomenon can be observed in Fig. \ref{fig.2-2}, which ensures that the instantaneous frequencies of all chirp subcarriers are confined within the range of $\left[0, \frac{1}{\Delta t}\right]$, guaranteeing a total bandwidth of $B=\frac{1}{\Delta t}$ approximately. Moreover, all chirp subcarriers share the same analog chirp rate of $2\tilde{c}_{1}$ but have a unique initial phase of $j2\pi c_{2} m^{2}$, and a uniform subcarrier spacing of $\Delta f = \frac{B}{N}$ is established among all chirp subcarriers. Furthermore, each chirp subcarrier can be viewed as a combination of $C$ “\textbf{\emph{subchirps}}", where each subchirp occupies a duration of $T_{\text{SC}}=\frac{T}{C}$ with its instantaneous frequency sweeping from zero to $B$ (for the $m$th  chirp subcarrier, $m \neq 0$, concatenating its head and tail effectively forms an intact subchirp)}.
	\label{obs1}
\end{observation}

\subsection{ISAC Channel Model}
\label{sec2-2}
As shown in Fig. \ref{fig2-1}, part of the AFDM-ISAC signal interacts with the communication channel and arrives at the communication receiver, while part of it reaches the targets and is reflected to a sensing receiver. Generally, both communication and sensing channels are modeled as DSC, whose impulse responses with delay $\tau$ and Doppler $\nu$ can be defined in the following way as \cite{bb24.08.28.1}
\begin{equation}
	h(\tau, \nu)=\sum_{i=1}^{P} h_{i} \delta\left(\tau-\tau_{i}\right) \delta\left(\nu -\nu_{i}\right),
	\label{eq2-4}
\end{equation}
where $P$ denotes the number of paths (targets), $h_{i}$, $\tau_{i}$, and $\nu_{i}$ denote the channel gain, delay shift, and Doppler shift of the $i$th path, respectively. For monostatic sensing, the round-trip delay and Doppler shift of the $i$th target are given by $\tau_{i}=\frac{2d_{i}}{v_{c}}$ and $\nu_{i}=\frac{2v_{i}f_{c}}{v_{c}}$, respectively, where $d_{i}$ and $v_{i}$ denote the distance and relative radial velocity of the target, and $v_{c}$ and $f_{c}$ represent the speed of light and the carrier frequency, respectively. In the bistatic sensing, these parameters become $\tau_{i}=\frac{d_{i,t}+d_{i,r}}{v_{c}}$ and $\nu_{i}=\frac{(v_{i,t}+v_{i,r})f_{c}}{v_{c}}$, where $d_{i,t}$ and $d_{i,r}$ are the distances from the transmitter and receiver to the target, and $v_{i,t}$ and $v_{i,r}$ are the corresponding relative radial velocities. Subsequently, the received AFDM-ISAC signal, for both the communication and sensing receivers, can be formulated as
\begin{align}
	r(t)&=\int_{-\infty}^{\infty} \int_{-\infty}^{\infty} h(\tau, \nu) s(t-\tau) e^{j 2 \pi \nu t} \text{d} \tau \text{d} \nu + w(t)  \notag\\
	&=\sum_{i=1}^{P} h_{i}  s(t-\tau_{i})e^{j  2 \pi\nu_{i}t}+w(t),  
	\label{eq24.01.09.2}
\end{align}
where $w(t)\sim \mathcal{C} \mathcal{N}\left(0, N_{0}\right)$ is the time-domain additive white Gaussian noise (AWGN) process.

\subsection{AFDM-ISAC Receivers}
\label{sec2-3}

\emph{\textbf{1) Communication receiver:}} After performing Nyquist sampling on $r(t)$ and discarding the CPP component, the received discrete-time communication signal $r_{\text{com}}[n]$ can be obtained, which is then demodulated via DAFT as
\begin{equation}
	y[m]= \frac{1}{\sqrt{N}}\sum_{n=0}^{N-1} r_{\text{com}}[n] \phi_{m}^{*}[n],  \quad m=0,1, \dots, N-1,
	\label{eq24.01.09.3}
\end{equation}
enabling DAFT-domain channel estimation \cite{23.10.18.1,bb25.01.25.1,bb25.01.25.2} and signal detection \cite{bb25.01.09.4,bb25.01.25.3}.

\emph{\textbf{2) Monostatic sensing receiver (MSR):}} At the MSR, the transmitted AFDM-ISAC signal $s(t)$ is fully known despite its randomness introduced by the random data symbols given that the MSR is collocated with the AFDM-ISAC transmitter. Therefore, one common operation is performing MF on the echo signal, i.e., the received AFDM-ISAC signal $r_{\text{MSR}}(t)$ with $s(t)$, which can be formulated as 
\begin{equation}
	R(\tau, \nu)= \int_{-\infty}^{\infty} r_{\text{MSR}}(t)s^{*}(t-\tau)e^{-j  2 \pi\nu t} \text{d}t.
	\label{eq24.01.09.4}
\end{equation}
Subsitituting  (\ref{eq24.01.09.2}) into (\ref{eq24.01.09.4}), we have
\begin{align}
	&R(\tau, \nu)  \notag\\
	&=\int_{-\infty}^{\infty} \left(\sum_{i=1}^{P} h_{i}  s(t-\tau_{i})e^{j2\pi\nu_{i}t} \right)s^{*}(t-\tau)e^{-j  2 \pi\nu t} \text{d}t + \tilde{w}(t)  \notag \\
	&=\sum_{i=1}^{P} \tilde{h}_{i}\int_{-\infty}^{\infty}s(t)s^{*}(t-(\tau-\tau_{i}))e^{-j2\pi(\nu-\nu_{i})t}\text{d}t + \tilde{w}(t) \notag \\
	&=\sum_{i=1}^{P} \tilde{h}_{i}A_{s,s}(\tau-\tau_{i}, \nu-\nu_{i}) + \tilde{w}(t),
	\label{eq24.01.09.5}
\end{align}
where $\tilde{w} (t) = \int_{-\infty}^{\infty}w(t)s^{*}(t-\tau)e^{-j  2 \pi\nu t} \text{d}t$ is the noise after MF, $\tilde{h}_{i}=h_{i}e^{-j2\pi(\nu-\nu_{i})\tau_{i}}$, $A_{s,s}(\tau, \nu)$ is the aperiodic AF with a definition given by
\begin{equation}
	A_{a,b}(\tau, \nu)\triangleq \int_{-\infty}^{\infty}a(t)b^{*}(t-\tau)e^{-j2\pi\nu t}\text{d}t.
	\label{eq24.01.09.6}
\end{equation}
In particular, AF is known as AAF and CAF when $b(t)=a(t)$ and $b(t)\neq a(t)$, respectively. Moreover, as shown in (\ref{eq24.01.09.5}), the output of MF of $r_{\text{MSR}}(t)$ and $s(t)$ is the superposition of $P$ scaled and shifted AAF of $s(t)$, where the shifts along the delay domain and the Doppler domain are exactly the delay shifts and the Doppler shifts associated with the $P$ sensing targets, respectively\footnote{Although the AF analysis in this paper focuses on a single-antenna setting, it can be readily extended to various multi-antenna configurations by incorporating the spatial domain into the AF through the use of appropriate spatial steering vectors.}. 

\emph{\textbf{3) Bistatic sensing receiver (BSR):}} The BSR has a different location from the AFDM-ISAC transmitter, which means only the pilot component of $s(t)$, i.e., one pilot chirp subcarrier, is known\footnote{Here, we assume that there is no fronthaul connection between the ISAC transmitter and the BSR; otherwise, the BSR could utilize the full signal $s(t)$ to perform MF as in (\ref{eq24.01.09.4}), rather than relying solely on its pilot component.}. Therefore, the BSR usually performs MF on the echo signal $r_{\text{BSR}}(t)$ with the pilot component of $s(t)$, which is deeply determined by the CAF between $s(t)$ and one chirp subcarrier.

The above discussion highlights the close connection between the MF operation in sensing and the AF. Building on these insights, we proceed with the analysis of the AFs of AFDM-ISAC signals in the following. For clarity, we focus on the AFs of CPP-free AFDM-ISAC signals, which comprise $\phi_{m}(t)$ in (\ref{eq2-3}). The length of the CPP is determined by the communication channel and is typically independent of the sensing channel. However, as we shall demonstrate later, all conclusions presented in this paper remain valid for CPP-appended AFDM-ISAC signals.

\section{AAF of AFDM Chirp Subcarriers}
\label{sec3}
As demonstrated in the IDAFT provided in (\ref{eq2-1}), each DAFT-domain symbol in $\mathbf{x}$ is associated with a chirp subcarrier, and every AFDM signal is essentially the combination of $N$ chirp subcarriers. Therefore, we are motivated to investigate the aperiodic AAF of the AFDM chirp subcarriers, which serves as the foundation for analyzing the AFs of AFDM frames in Sec. \ref{sec5}.
\begin{figure*}[htbp]
	\centering
	\includegraphics[width=0.98\textwidth,height=0.47\textwidth]{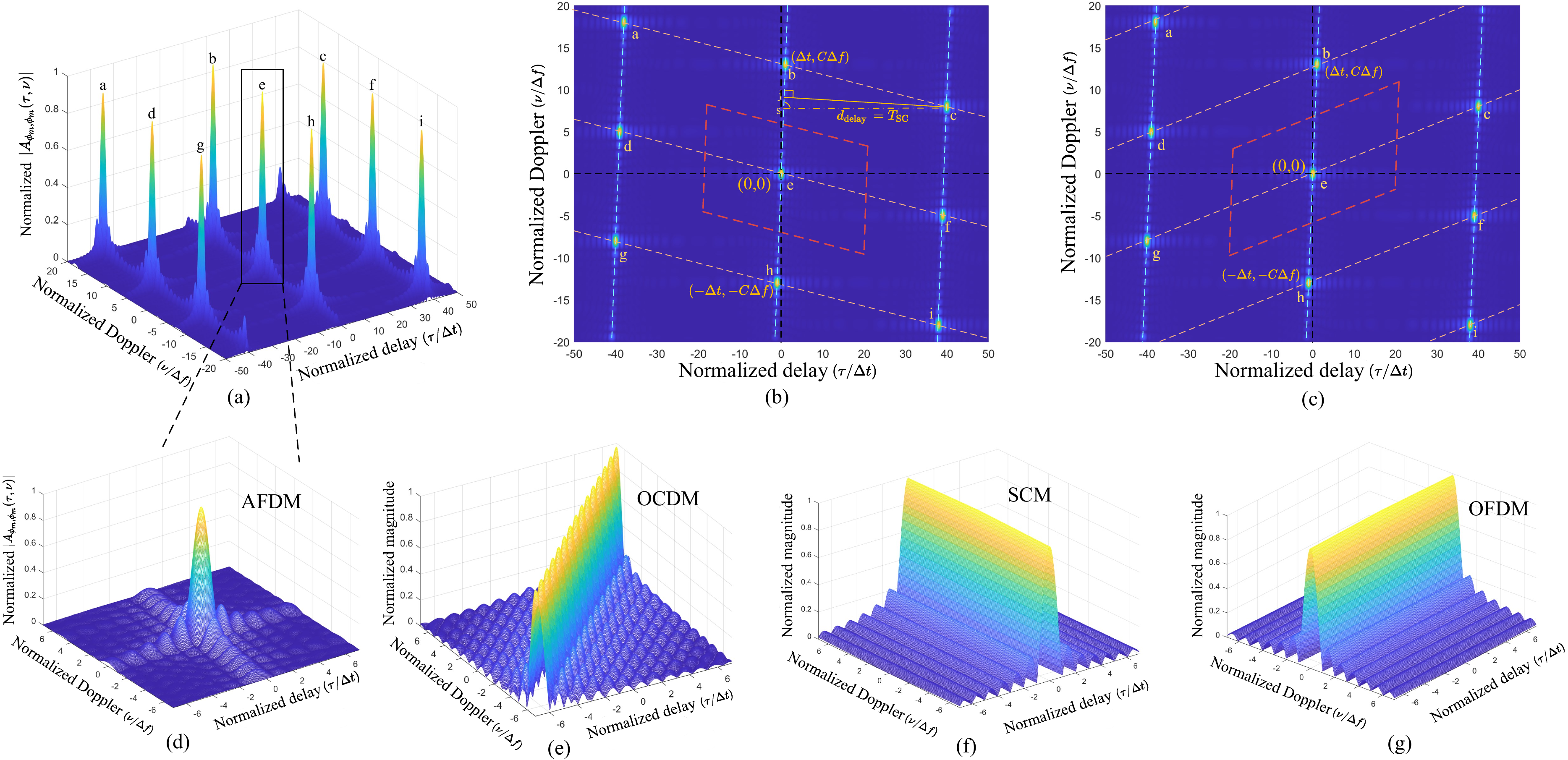}
	\vspace{-0.35cm}
	\caption{Subcarrier/symbol AAFs of different waveforms, $N=512$, $C=13$, $m=47$: (a) AAF of the AFDM chirp subcarrier; (b) Planform of the AFDM chirp subcarrier AAF; (c) Planform of the AFDM chirp subcarrier AAF with another type of unambiguity parallelogram; (d) Zoomed AAF of the AFDM chirp subcarrier; (e) AAF of the OCDM chirp subcarrier; (f) AAF of the SCM symbol; (g) AAF of the OFDM subcarrier.}
	\vspace{-0.41cm}
	\label{fig.3-1}
\end{figure*}

\subsection{Derivation of the AAF of AFDM Chirp Subcarriers}
We first derive the AAF of the $m$th chirp subcarrier, which can be formulated as
\begin{equation}
	A_{\phi_{m},\phi_{m}}(\tau, \nu)= \int_{-\infty}^{\infty}
	\phi_{m}(t)
	\phi^{*}_{m}(t-\tau)e^{-j2\pi\nu t}\text{d}t.
	\label{eq24.01.11.1}
\end{equation}
 For ease of derivation, we rewrite $\phi_{m}(t)$ in (\ref{eq2-3}) as
\begin{equation} 
	\phi_{m}(t) = e^{j 2 \pi\left(c_{2}m^{2}+\tilde{c}_{1}t^{2}+\frac{m}{T}t-\frac{q_{m}(t)}{\Delta t}t\right)}, \ \ \
	0 \le t < T,
	\label{eq24.01.11.2}
\end{equation}
where
\begin{equation} 
	q_{m}(t) = q, \quad
	t_{m,q} \le t < t_{m,q+1}
	\label{eq24.01.11.3}
\end{equation}
is a step function that indicates the intentional spectrum wrapping in chirp subcarriers. Substituting (\ref{eq24.01.11.2}) into (\ref{eq24.01.11.1}), we have (\ref{eq24.01.11.4}) (shown at the top of the next page), which shows no relevance to  $c_{2}$.
\begin{figure*}
		\begin{align}
		A_{\phi_{m},\phi_{m}}(\tau, \nu)  &= \int_{-\infty}^{\infty} e^{j 2 \pi\left(c_{2}m^{2}+\tilde{c}_{1}t^{2}+\frac{m}{T}t-\frac{q_{m}(t)}{\Delta t}t\right)}
		e^{-j 2 \pi\left[c_{2}m^{2}+\tilde{c}_{1}(t-\tau)^{2}+\frac{m}{T}(t-\tau)-\frac{q_{m}(t-\tau)}{\Delta t}(t-\tau)\right]}e^{-j2\pi\nu t}\text{d}t \notag\\
		&=\int_{-\infty}^{\infty}
		\underbrace{
			e^{j 2 \pi
		\left(2\tilde{c}_{1}\tau 
				-\frac{q_{m}(t)-q_{m}(t-\tau)}{\Delta t}-\nu \right)t}}_{\mathcal{F}^{\tau, \nu}_{m,\tilde{c}_{1}}(t)} 
		\underbrace{
			e^{j 2 \pi\left(\frac{m}{T}\tau-\frac{q_{m}(t-\tau)}{\Delta t}\tau-\tilde{c}_{1}\tau^{2}\right)
		}}_{\mathcal{P}^{\tau}_{m,\tilde{c}_{1}}(t)} \text{d}t.
		\label{eq24.01.11.4}	
	\end{align}
\hrulefill
\end{figure*}
Note that owing to the piecewise property of AFDM chirp subcarriers, the integral in (\ref{eq24.01.11.4}) should also be calculated in a piecewise way correspondingly, whose exact expression shows a slight difference with respect to the value of $\tau$. Due to space limitations, we only demonstrate one representative case, from which the remaining cases can be obtained with minor modifications. To proceed, we define an indicator function $\tilde{q}_{m}(\tau)$ as 
\begin{equation} 
	\tilde{q}_{m}(\tau) = q+1, \quad
	t_{m,q} \le \tau < t_{m,q+1}
	\label{eq24.01.11.5}
\end{equation}
and assume that
\begin{equation} 
	\tau \in \big[kT_{\text{SC}},  \min\{(T_{\text{SC}}-t_{m,1}),t_{m,1}\}+kT_{\text{SC}}\big],
	\label{eq24.01.11.6}
\end{equation}
where $k = 0,1,\dots, C-1$. Then, (\ref{eq24.01.11.4}) can be derived as  (\ref{eq24.01.12.1}), which is based on the fact that $\phi_{m}(t)$ has finite support over the interval $[0, T]$ and it can be further decomposed by substituting (\ref{eq24.01.08.1}) and (\ref{eq24.01.11.3}), as shown in (\ref{eq24.01.12.2}). 
\begin{figure*}
	\begin{align}
		&A_{\phi_{m},\phi_{m}}(\tau, \nu)\overset{(\ref{eq24.01.11.6})}{=} \int_{\tau}^{T} 
		e^{j 2 \pi
			\left(2\tilde{c}_{1}\tau 
			-\frac{q_{m}(t)-q_{m}(t-\tau)}{\Delta t}-\nu \right)t}
		e^{j 2 \pi\left(\frac{m}{T}-\frac{q_{m}(t-\tau)}{\Delta t}-\tilde{c}_{1}\tau\right)\tau
		} \text{d}t \label{eq24.01.12.1}\\
		&\overset{(\ref{eq24.01.08.1},\ref{eq24.01.11.3})}{=}
		\sum_{i=0}^{C-\tilde{q}_{m}(\tau)}
		\Bigg(
		\underbrace{
			\int_{t_{m,i}+\tau}^{t_{m,i+\tilde{q}_{m}(\tau)}} 
			e^{j 2 \pi
				\left(2\tilde{c}_{1}\tau 
				-\frac{\tilde{q}_{m}(\tau)-1}{\Delta t}-\nu \right)t} 
			e^{j 2 \pi\left(\frac{m}{T}-\frac{i}{\Delta t}-\tilde{c}_{1}\tau\right)\tau
			} \text{d}t}_{\mathcal{S}_{a}^{i}} 
		+\underbrace{
			\int_{t_{m,i+\tilde{q}_{m}(\tau)}}^{t_{m,i+1}+\tau} 
			e^{j 2 \pi
				\left(2\tilde{c}_{1}\tau 
				-\frac{\tilde{q}_{m}(\tau)}{\Delta t}-\nu \right)t} 
			e^{j 2 \pi\left(\frac{m}{T}-\frac{i}{\Delta t}-\tilde{c}_{1}\tau\right)\tau
			} \text{d}t}_{\mathcal{S}_{b}^{i}}
		\Bigg)  \notag\\
		&\qquad \qquad \qquad \qquad \qquad \qquad \qquad \qquad \qquad  \qquad  
		+\underbrace{\int_{t_{m,C-\tilde{q}_{m}(\tau)+1}+\tau}^{T} e^{j 2 \pi
				\left(2\tilde{c}_{1}\tau 
				-\frac{\tilde{q}_{m}(\tau)-1}{\Delta t}-\nu \right)t} 
			e^{j 2 \pi\left(\frac{m}{T}-\frac{C-\tilde{q}_{m}(\tau)+1}{\Delta t}-\tilde{c}_{1}\tau\right)\tau
			} \text{d}t}_{\mathcal{S}_{a}^{C-\tilde{q}_{m}(\tau)+1}} \label{eq24.01.12.2}\\
			&=
			e^{j 2 \pi \left(\frac{m\tau}{T}-\tilde{c}_{1}\tau^{2}
			\right)}
		\Bigg[
		\sum_{i=0}^{C-\tilde{q}_{m}(\tau)}	
		e^{-j 2 \pi\frac{i}{\Delta t}\tau} 
		\Bigg( 
		\mathcal{I}\left(\mathcal{E}^{\tau, \nu}_{m,\tilde{c}_{1}}, \left(t_{m,i}+\tau\right), \left(t_{m,i+\tilde{q}_{m}(\tau)}\right)\right)
		+
		\mathcal{I}\left(\tilde{\mathcal{E}}^{\tau, \nu}_{m,\tilde{c}_{1}}, \left(t_{m,i+\tilde{q}_{m}(\tau)}\right), \left(t_{m,i+1}+\tau\right)\right)
		\Bigg)  \notag\\
		&\qquad \qquad \qquad \qquad \qquad \qquad \qquad \qquad \qquad  \qquad  \qquad \qquad \quad 
		+ e^{-j 2 \pi\frac{C-\tilde{q}_{m}(\tau)+1}{\Delta t}\tau
		}
		\mathcal{I}\left(\mathcal{E}^{\tau, \nu}_{m,\tilde{c}_{1}}, \left(t_{m,C-\tilde{q}_{m}(\tau)+1}+\tau\right), T\right)	\Bigg]. 
		\label{eq24.01.11.7}
	\end{align}
	\hrulefill
\end{figure*}
Finally, after some algebraic operations, we arrive at (\ref{eq24.01.11.7}), where the integral indicator
\begin{align}
	\mathcal{I}\left(c, t_{1}, t_{2}\right)  \triangleq \int_{t_{1}}^{t_{2}}  e^{j2\pi ct}\text{d}t= \left\{\begin{array}{ll}
	\frac{e^{j2\pi ct_{2}}-e^{j2\pi ct_{1}}}{j2\pi c}, &c \neq 0 \\
	t_{2}-t_{1},& c=0 
\end{array}\right. 
	\label{eq24.01.12.3}
\end{align}
and 
\begin{equation}
\mathcal{E}^{\tau, \nu}_{m,\tilde{c}_{1}}=2\tilde{c}_{1}\tau 
-\frac{\tilde{q}_{m}(\tau)-1}{\Delta t}-\nu, \ \tilde{\mathcal{E}}^{\tau, \nu}_{m,\tilde{c}_{1}}=\mathcal{E}^{\tau, \nu}_{m,\tilde{c}_{1}}-\frac{1}{\Delta t}.
\label{eq24.01.15.13}
\end{equation}
Based on the above-derived results, we visualize $A_{\phi_{m},\phi_{m}}(\tau, \nu)$ in Fig. \ref{fig.3-1} and demonstrate its prominent properties in the following. 
\begin{figure}[tbp]
	\centering
	\includegraphics[width=0.47\textwidth,height=0.381\textwidth]{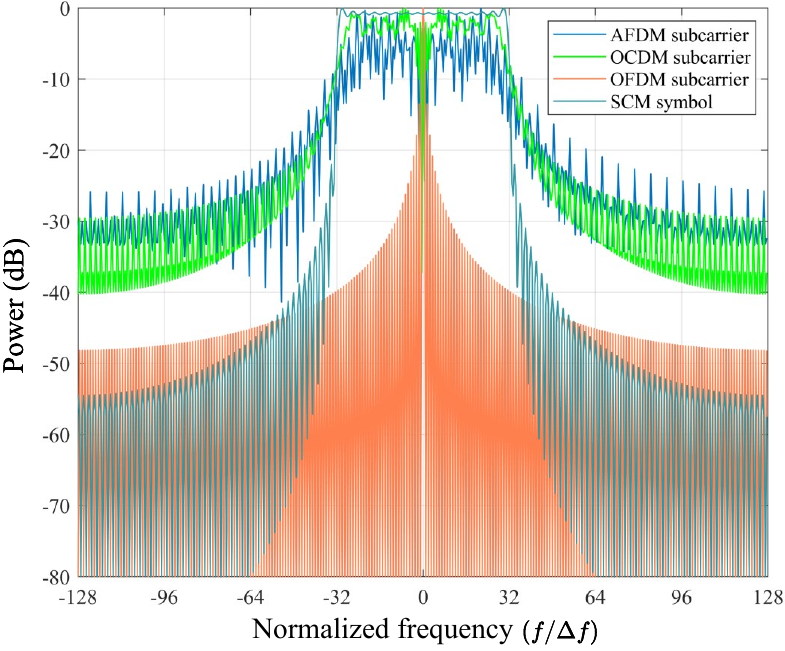}
	\caption{Comparison of normalized power spectral density (PSD) for a single subcarrier/symbol in OCDM, OFDM, SCM, and AFDM ($T\approx 66.67 \mu$s,  $N=64$, $C =3$, and  $c_{2} =\sqrt{2}$).}
	\label{fig.3-2}
\end{figure}

\subsection{Properties of the AAF of AFDM Chirp Subcarriers}
\label{sec3-2}
\begin{property}[“Spike-like” local property] 
	\textup{ The AAF of AFDM chirp subcarriers exhibit a “spike-like” pulse locally, whose mainlobe widths are approximately   $\Delta t =\frac{1}{B}$ and $\Delta f =\frac{1}{T}$ along the delay and Doppler dimensions, respectively, as shown in Fig. \ref{fig.3-1}(d)}.
\end{property}

The “spike-like” local property of the AAF of AFDM chirp subcarriers can be attributed to their time-frequency spanning nature. This characteristic allows them to fully occupy the time and frequency resource of an AFDM frame, providing delay perceptibility with a resolution of $\frac{1}{B}$ and Doppler perceptibility with a resolution of $\frac{1}{T}$ simultaneously. This is a significant difference that distinguishes AFDM from conventional SCM and OFDM, where each SCM symbol occupies a bandwidth of $B_{\text{SCM}}^{\text{symbol}}=B$ and has a duration of $T_{\text{SCM}}^{\text{symbol}}=\frac{2T}{N}$, while an OFDM subcarrier occupies a bandwidth of $B_{\text{OFDM}}^{\text{subcarrier}}=\frac{2B}{N}$ and has a duration of $T_{\text{OFDM}}^{\text{subcarrier}}=T$ approximately, as shown in Fig. \ref{fig.3-2}. Furthermore, the AAFs of a SCM symbol and an OFDM subcarrier are provided in Fig. \ref{fig.3-1}(f) and Fig. \ref{fig.3-1}(g), respectively. We can see that the delay mainlobe width of the AAF of a SCM symbol is $\frac{1}{B}$, while its Doppler mainlobe width is expected to be $B$, which means a SCM symbol only possesses delay perceptibility. Similarly, the Doppler mainlobe width of the AAF of an OFDM subcarrier is $\frac{1}{T}$ approximately, while its delay mainlobe width is expected to be $T$, which means that an OFDM subcarrier only possesses Doppler perceptibility \cite{bb25.01.21.2}. This clearly demonstrates the advantage of AFDM chirp subcarriers over SCM symbols and OFDM subcarriers in radar sensing. Note that due to the spectrum wrapping phenomenon in AFDM chirp subcarriers, as discussed in Observation \ref{obs1}, AFDM exhibits relatively high out-of-band emission (OOBE), as shown in Fig. \ref{fig.3-2}. This requires stricter spectrum management in practice, placing higher demands on digital-to-analog filters.

\begin{property}[“Periodic-like” global property] 
	\textup{The AAF of AFDM chirp subcarriers exhibit “periodic-like” pulses globally along the rotated Doppler and delay dimensions, as shown in Fig. \ref{fig.3-1}(b)(c)}. 
\end{property}

The “periodic-like” global property can be explained by considering the physical essence of the AAF in (\ref{eq24.01.11.1}) and the subchirp splicing feature of AFDM chirp subcarriers, as summarized in Observation 1. Specifically, the AAF quantifies the degree of correlation between $\phi_{m}(t)$ and $\phi_{m}(t-\tau)e^{j2\pi\nu t}$. For AFDM chirp subcarriers, which are composed of multiple subchirps, the strong correlation requirement leads to the following two conditions:

\textbf{1) \emph{Frequency alignment (FA)}} between the matched subchirps from $\phi_{m}(t)$ and $\phi_{m}(t-\tau)e^{j2\pi\nu t}$.  Fig. \ref{fig.3-3}(a) and Fig. \ref{fig.3-3}(c) illustrate the two cases of frequency alignment for a given delay shift. For ease of demonstration, we define the frequency difference indicator as
\begin{equation}
\mathcal{F}^{\tau, \nu}_{m,\tilde{c}_{1}}(t) =e^{j 2 \pi
	\left(2\tilde{c}_{1}\tau 
	-\frac{q_{m}(t)-q_{m}(t-\tau)}{\Delta t}-\nu \right)t},
\label{eq24.01.15.1}
\end{equation}
which indicates the instant frequency difference between $\phi_{m}(t)$ and $\phi_{m}(t-\tau)e^{j2\pi\nu t}$ at time $t$. Then, the FA condition requires that $\mathcal{F}^{\tau, \nu}_{m,\tilde{c}_{1}}(t)=1$, which means that for a given delay shift $\tau$, frequency alignment condition can be satisfied at time $t$ when
\begin{equation}
 \nu=2\tilde{c}_{1}\tau 
-\frac{q_{m}(t)-q_{m}(t-\tau)}{\Delta t},
	\label{eq24.01.14.1}
\end{equation}
which takes two different values for the two cases shown in Fig. \ref{fig.3-3}(a) and Fig. \ref{fig.3-3}(c), respectively. We can observe that, in each case, a portion of the matched subchirps from $\phi_{m}(t)$ and $\phi_{m}(t-\tau)e^{j2\pi\nu t}$ are overlapped with aligned instantaneous frequencies, while the remaining portions exhibit misaligned instantaneous frequencies. Insightfully, we associate the two parts of the matched subchirps with the decomposed form of the AAF by labeling them as follows: $\mathcal{S}_{a}^{i}$ ($i=0,1,\dots,C-\tilde{q}_{m}(\tau)+1$) and $\mathcal{S}_{b}^{i}$ ($i=0,1,\dots,C-\tilde{q}_{m}(\tau)$) in (\ref{eq24.01.12.2}), respectively, and identify them in Fig. \ref{fig.3-3}(a) and Fig. \ref{fig.3-3}(c). Moreover, let $\mathcal{T}_{\mathcal{S}_{a}^{i}}$ and $\mathcal{T}_{\mathcal{S}_{b}^{i}}$ represent the time period occupied by  $\mathcal{S}_{a}^{i}$ and $\mathcal{S}_{b}^{i}$, respectively, and
\begin{align}
	\mathcal{T}_{\mathcal{S}_{a}} &\triangleq \mathcal{T}_{\mathcal{S}_{a}^{0}}\cup\mathcal{T}_{\mathcal{S}_{a}^{1}}\cup\dots\cup\mathcal{T}_{\mathcal{S}_{a}^{C-\tilde{q}_{m}(\tau)+1}}\label{eq24.01.15.10},\\
	\mathcal{T}_{\mathcal{S}_{b}} &\triangleq \mathcal{T}_{\mathcal{S}_{b}^{0}}\cup\mathcal{T}_{\mathcal{S}_{b}^{1}}\cup\dots\cup\mathcal{T}_{\mathcal{S}_{b}^{C-\tilde{q}_{m}(\tau)}}.
	\label{eq24.01.15.6}
\end{align}
Then, for the case shown in Fig. \ref{fig.3-3}(a), the frequency-aligned parts are $t \in \mathcal{T}_{\mathcal{S}_{a}}$, where
\begin{equation}
	q_{m}(t)-q_{m}(t-\tau) = \tilde{q}_{m}(\tau)-1,  \quad t\in \mathcal{T}_{\mathcal{S}_{a}}
	\label{eq24.01.15.5}
\end{equation}
holds and hence (\ref{eq24.01.14.1}) turns into
\begin{equation}
	\nu=2\tilde{c}_{1}\tau 
	-\frac{\tilde{q}_{m}(\tau)-1}{\Delta t}, \quad \text{for case Fig. \ref{fig.3-3}(a)}.  
	\label{eq24.01.15.9}
\end{equation}
The frequency-misaligned parts are $t\in\mathcal{T}_{\mathcal{S}_{b}}$, where the instant frequency difference is $-\frac{1}{\Delta t}$, i.e., 
\begin{equation}
	\mathcal{F}^{\tau, \nu}_{m,\tilde{c}_{1}}(t) =e^{-j 2 \pi \frac{t}{\Delta t}}, \quad  \text{for case Fig. \ref{fig.3-3}(a)}, \ t\in \mathcal{T}_{\mathcal{S}_{b}}.
	\label{eq24.01.15.14}
\end{equation}
Similarly, for the case in Fig. \ref{fig.3-3}(c), the frequency-aligned parts are $t \in \mathcal{T}_{\mathcal{S}_{b}}$, where
\begin{equation}
	q_{m}(t)-q_{m}(t-\tau) = \tilde{q}_{m}(\tau),  \quad t\in \mathcal{T}_{\mathcal{S}_{b}}
	\label{eq24.01.15.11}
\end{equation}
holds and hence (\ref{eq24.01.14.1}) turns into
\begin{equation}
	\nu=2\tilde{c}_{1}\tau 
	-\frac{\tilde{q}_{m}(\tau)}{\Delta t},\quad \text{for case Fig. \ref{fig.3-3}(c)}. 
	\label{eq24.01.15.12}
\end{equation}
Correspondingly, the frequency-misaligned parts are $t\in\mathcal{T}_{\mathcal{S}_{a}}$, where the instant frequency difference is $\frac{1}{\Delta t}$, i.e., 
\begin{equation}
	\mathcal{F}^{\tau, \nu}_{m,\tilde{c}_{1}}(t) =e^{j 2 \pi
		\frac{t}{\Delta t}}, \quad  \text{for case Fig. \ref{fig.3-3}(c)}, \ t\in \mathcal{T}_{\mathcal{S}_{a}}.
	\label{eq24.01.15.15}
\end{equation}

\begin{figure}[tbp]
	\centering
	\includegraphics[width=0.45\textwidth,height=0.44\textwidth]{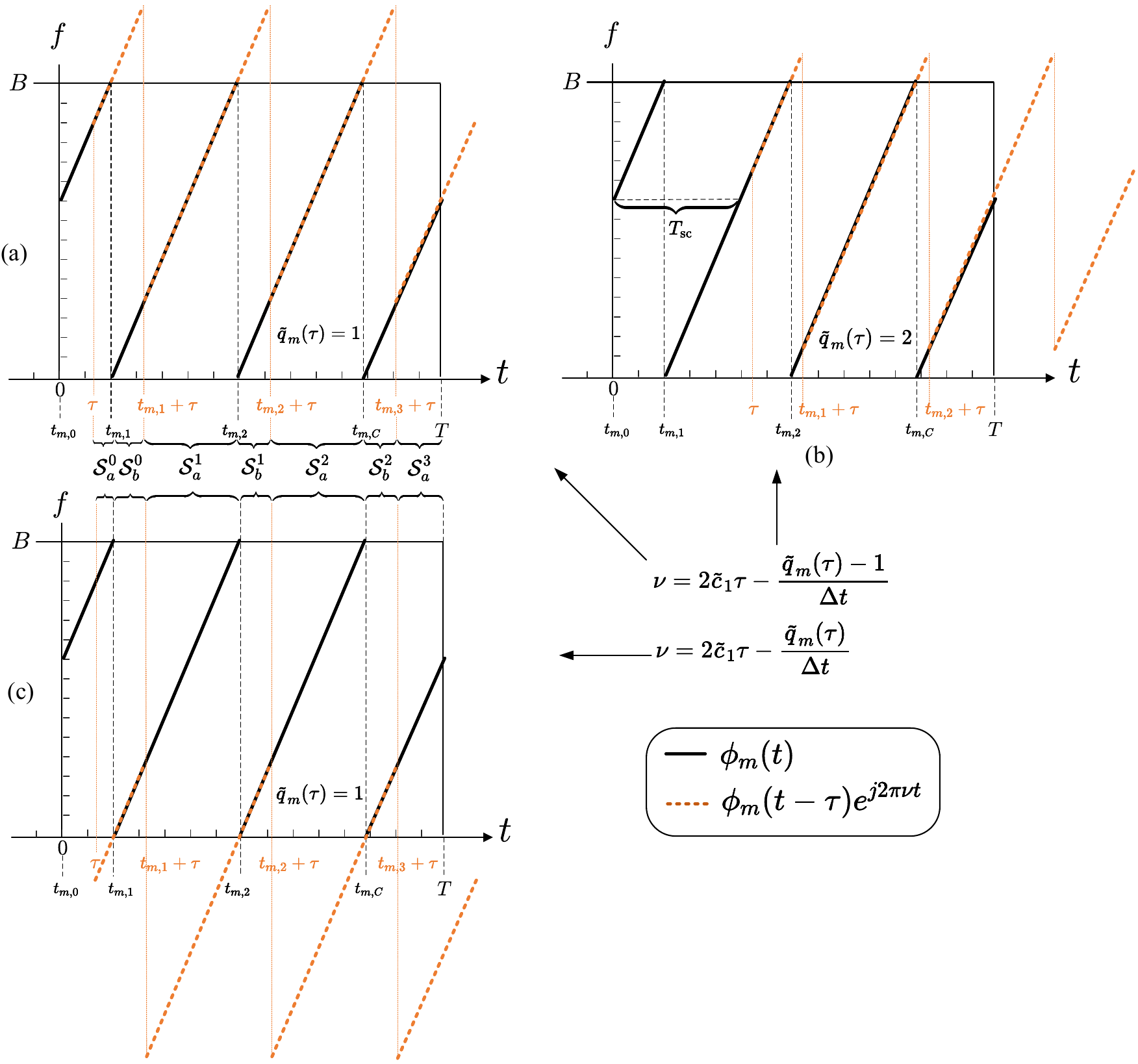}
	\caption{Illustration of the integral in (\ref{eq24.01.11.4}) and (\ref{eq24.01.12.2}), $C=3$.}
	\label{fig.3-3}
\end{figure}

\textbf{2) \emph{Residue phase difference alignment (RPDA)}} among the matched subchirp pairs from $\phi_{m}(t)$ and $\phi_{m}(t-\tau)e^{j2\pi\nu t}$. We further define the residue phase difference indicator as
\begin{equation}
	\mathcal{P}^{\tau}_{m,\tilde{c}_{1}}(t) =e^{j 2 \pi\left(\frac{m}{T}\tau-\frac{q_{m}(t-\tau)}{\Delta t}\tau-\tilde{c}_{1}\tau^{2}\right)},
	\label{eq24.01.15.2}
\end{equation}
which indicates the instant residue phase difference between $\phi_{m}(t)$ and $\phi_{m}(t-\tau)e^{j2\pi\nu t}$ at time $t$. Then,
the RPDA condition suggests that the frequency-aligned parts of all matched subchirp pairs share the same instant residue phase difference, i.e.,  $\mathcal{P}^{\tau}_{m,\tilde{c}_{1}}(t)$ maintains the same value when $t\in \mathcal{T}_{\mathcal{S}_{a}}$ for case Fig. \ref{fig.3-3}(a) or $t\in \mathcal{T}_{\mathcal{S}_{b}}$ for case Fig. \ref{fig.3-3}(c). This means that the influence of variable $t$ in $\frac{q_{m}(t-\tau)}{\Delta t}\tau$ should be eliminated from $\mathcal{P}^{\tau}_{m,\tilde{c}_{1}}(t)$, which can be achieved by letting $\frac{q_{m}(t-\tau)}{\Delta t}\tau$ be an integer. According to (\ref{eq24.01.11.3}),  $q_{m}(t-\tau)$ is always an integer, meaning that the elimination of $t$ from $\mathcal{P}^{\tau}_{m,\tilde{c}_{1}}(t)$ can only hold when
\begin{equation}
	\tau = k \Delta t, \quad k\in\mathbb{Z}.
	\label{eq24.01.14.2}
\end{equation}

Therefore, when (\ref{eq24.01.14.1}) and (\ref{eq24.01.14.2}) are satisfied, the frequency alignment and the residue phase difference alignment are achieved at the frequency-aligned part of the matched subchirps simultaneously. Subsequently, for the case in Fig. \ref{fig.3-3}(a), substituting   (\ref{eq24.01.15.9}) and (\ref{eq24.01.14.2}) into $\mathcal{S}_{a}^{i}$ in (\ref{eq24.01.12.2}) arrives at
\begin{align}
	&\mathcal{S}_{a}^{i} =\int_{t_{m,i}+k \Delta t}^{t_{m,i+\tilde{q}_{m}(k \Delta t)}}
	e^{j 2 \pi\left(\frac{m}{T}-\tilde{c}_{1}k \Delta t\right)k \Delta t
	} \text{d}t \notag \\
&=e^{j 2 \pi\left(\frac{m}{T}-\tilde{c}_{1}k \Delta t\right)k \Delta t} [(t_{m,i+\tilde{q}_{m}(k \Delta t)})-(t_{m,i}+k \Delta t)],
	\label{eq24.01.15.3}
\end{align}
which yields
\begin{align}
	&\sum_{i=0}^{C-\tilde{q}_{m}(\tau)+1}\mathcal{S}_{a}^{i} = e^{j 2 \pi\left(\frac{m}{T}-\tilde{c}_{1}k \Delta t\right)k \Delta t} \notag \\
	& \qquad \quad \ \times \sum_{i=0}^{C-\tilde{q}_{m}(\tau)+1}[(t_{m,i+\tilde{q}_{m}(k \Delta t)})-(t_{m,i}+k \Delta t)],
	\label{eq24.01.15.4}
\end{align}
indicating energy accumulation, which in turn boosts a pulse in $A_{\phi_{m},\phi_{m}}(\tau, \nu)$. In contrast, the frequency-misaligned parts can be calculated by substituting   (\ref{eq24.01.15.14}) and (\ref{eq24.01.14.2}) into $\mathcal{S}_{b}^{i}$ in (\ref{eq24.01.12.2}) as
\begin{align}
	\mathcal{S}_{b}^{i} &=\int_{t_{m,i}+\Delta t}^{t_{m,i+\tilde{q}_{m}(\Delta t)}} e^{-j 2 \pi \frac{t}{\Delta t}}
	e^{j 2 \pi\left(\frac{m}{T}-\tilde{c}_{1}k \Delta t\right)k \Delta t} \text{d}t,
	\label{eq24.01.16.1}
\end{align}
which has a small magnitude compared to $\mathcal{S}_{a}^{i}$ given its integrand is a complex-exponential function with a period of $\Delta t$. A similar result can be obtained for case Fig. \ref{fig.3-3}(c) by substituting   (\ref{eq24.01.15.12}) and (\ref{eq24.01.14.2}) into $\mathcal{S}_{b}^{i}$ in (\ref{eq24.01.12.2}), where an energy accumulation process will be generated by $\sum_{i=0}^{C-\tilde{q}_{m}(\tau)}\mathcal{S}_{b}^{i}$.

\begin{remark} 
	\textup{We can observe from (\ref{eq24.01.15.9}) and (\ref{eq24.01.15.12}) that the relationship between $\nu$ and $\tau$ under the FA and RPDA conditions is a linear function with a slope of $2\tilde{c}_{1}$. This results in pulse “periodicity” along the rotated Doppler dimension with a slope $2\tilde{c}_{1}$ in $A_{\phi_{m},\phi_{m}}(\tau, \nu)$, where two adjacent pulses have a relative delay distance of $\Delta t$ and a relative Doppler distance of $2\tilde{c}_{1}\Delta t=C\Delta f$, as illustrated in Fig. \ref{fig.3-1}(b). }
\end{remark} 

\begin{remark} 
	\textup{The FA and RPDA conditions can also be satisfied by properly sliding $\phi_{m}(t-\tau)e^{j2\pi\nu t}$ in Fig. \ref{fig.3-3}(a) to establish a new group of matched subchirp pairs with $\phi_{m}(t)$, as an example shown in Fig. \ref{fig.3-3}(b). This gives rise to multiple parallel lines with “periodic-like” pulses along the rotated Doppler dimension in $A_{\phi_{m},\phi_{m}}(\tau, \nu)$, as indicated by the light blue dashed lines in Fig. \ref{fig.3-1}(b)(c), where the delay-dimension distance between two adjacent parallel lines is $d_{\text{delay}}=T_{\text{SC}}$ exactly. Moreover, when $\tau = T_{\text{SC}}$ does not satisfy (\ref{eq24.01.14.2}), the pulses localized on different lines are distributed not along the delay dimension but along a rotated delay dimension, as indicated by the light orange dashed lines in Fig. \ref{fig.3-1}(b)(c).}
\end{remark}  
Consequently, a parallelogram can be virtually outlined by any four adjacent pulses in the DD plane of the AAF, either in the form of Fig. \ref{fig.3-1}(b) or Fig. \ref{fig.3-1}(c), effectively forming a virtual grid. Building on this analysis, we arrive at the following interesting conclusion, stated as Corollary 1.
 
\begin{corollary} 
 	The areas of the parallelograms in the AAF of AFDM chirp subcarriers are equal to one.
\end{corollary} 
\begin{proof}
	See Appendix \ref{APP1}.
\end{proof}

Corollary 1 suggests a tradeoff between the ranges of delay estimation and Doppler estimation when using an AFDM chirp subcarrier for unambiguous sensing, a topic that will be further discussed in the next subsection. 

\begin{remark} 
	\textup{The characteristics of aperiodic AAF described above hold for all AFDM chirp subcarriers, regardless of whether a CPP is appended. This is because the time-frequency distribution of all AFDM chirp subcarriers with finite support follows the same subchirp splicing properties outlined in Observation 1, independent of the presence of a CPP.}
\end{remark} 

\subsection{Unambiguous Sensing with AFDM Chirp Subcarriers}
Due to the “periodic-like” global property discussed above and the MF process introduced in Sec. \ref{sec2-3}, ambiguous sensing may arise when the pulses in $R(\tau, \nu)$ cannot be accurately associated with their corresponding target's AAF. To address this issue, we analyze the conditions necessary for unambiguous sensing in Proposition 1.

\begin{proposition} 
	Unambiguous sensing can be achieved when the delay and Doppler shifts of all targets in the sensing channel are confined within a region that matches the shape of the parallelogram formed by the AAF of AFDM chirp subcarriers.
	\label{prop1}
\end{proposition} 
\begin{proof}
	The “periodic-like” pulse distribution of the AAF of AFDM chirp subcarriers essentially allocates a detection region for each localized pulse, which can be acquired by integrally shifting the virtual grids in Fig. \ref{fig.3-1}(b)(c). The detection region of the origin pulse is illustrated by the red dashed lines. According to (\ref{eq24.01.09.5}), when transmitting $\phi_{m}(t)$, the output of the MF process $R(\tau, \nu)$ is a superposition of multiple scaled and shifted instances of $A_{\phi_{m},\phi_{m}}(\tau, \nu)$, where each localized pulse is uniformly shifted by the delay and Doppler shifts of the corresponding sensing target. Subsequently, when all targets fall within the detection region of the original pulse, all pulses in $A_{\phi_{m},\phi_{m}}(\tau, \nu)$ will uniformly spread only within their own detection region in $R(\tau, \nu)$, avoiding unambiguous sensing perfectly.  Moreover, the detection regions have an identical shape as the parallelogram of the AAF of AFDM chirp subcarriers, and we hence refer to them as \textbf{\emph{unambiguity parallelogram}} and refer to Proposition 1 as \textbf{\emph{unambiguous sensing condition}} for the AFDM chirp subcarrier. This completes the proof.
\end{proof}

\begin{remark} 
	\textup{The above proposition suggests that, for a given sensing channel and time-frequency resource, one should choose $C$ in AFDM-ISAC systems sufficiently large by tuning parameter $c_{1}$ to satisfy the unambiguous sensing condition. Take a counter-example that will be illustrated in Sec. \ref{sec6}, i.e., an OCDM chirp subcarrier, corresponding to $C=1$, has a tabular unambiguity parallelogram with insufficient Doppler sensing range and hence suffers from ambiguous sensing. Note that $C$ should be set as small as possible to reduce the channel estimation and multi-user access overhead, which is in direct proportion to $c_{1}$. This insight is in line with the discussion on the setting of $c_{1}$ to guarantee full DD representation in the DAFT domain in \cite{bb24.08.27.2}.}
\end{remark}

\section{CAF between Two AFDM Chirp Subcarriers}
\label{sec4}
\begin{figure}[tbp]
	\centering
	\includegraphics[width=0.45\textwidth,height=0.285\textwidth]{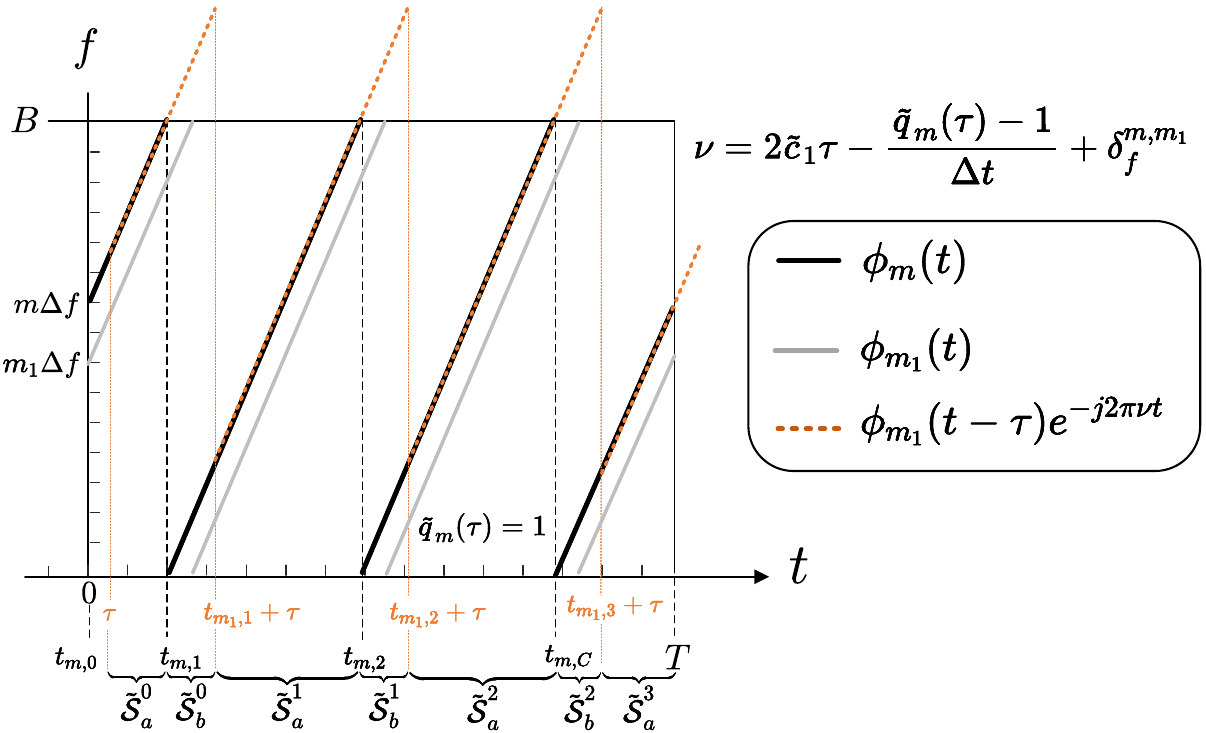}
	\caption{Illustration of the integral in (\ref{eq24.01.19.6}), $C=3$.}
	\label{fig.3-4}
\end{figure}

\begin{figure}[tbp]
	\centering
	\includegraphics[width=0.445\textwidth,height=0.36\textwidth]{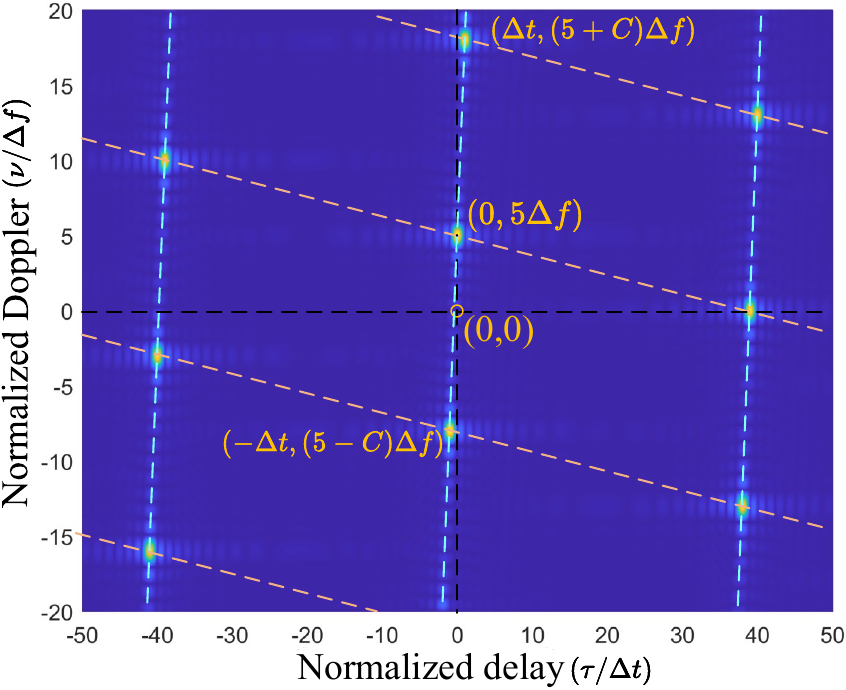}
	\caption{Planform of $\left|	A_{\phi_{m},\phi_{m_{1}}}(\tau, \nu)\right|$,  $N=512$, $C=13$, $c_{2} =\sqrt{2}$, $m=445$, $m_{1}=440$   ($\delta_{f}^{m,m_{1}}=5\Delta f$).}
	\label{fig.3-5}
\end{figure}
In this section, we investigate the CAF between two different AFDM chirp subcarriers. The aperiodic CAF between the $m$th and $m_{1}$th subcarriers can be formulated as 
\begin{align}
	A_{\phi_{m},\phi_{m_{1}}}(\tau, \nu)= \int_{-\infty}^{\infty}
	\phi_{m}(t)
	\phi^{*}_{m_{1}}(t-\tau)e^{-j2\pi\nu t}\text{d}t \label{eq24.01.19.5}
\end{align}
and (\ref{eq24.01.19.4}) (shown at the top of the next page). Similar to the AAF of the AFDM chirp subcarriers, the integral in (\ref{eq24.01.19.4}) should also be calculated in a piecewise way, whose exact expression exhibits slight variations depending on the value of $\tau$ and the relationship between $m$ and $m_{1}$. Without loss of generality, we present one representative case in detail. Assume that 
\begin{equation}
	\tau \in \big[kT_{\text{SC}}, \ \min\{(T_{\text{SC}}-t_{m_{1},1}),t_{m,1}\}+kT_{\text{SC}}\big],
	\label{eq24.01.20.1}
\end{equation}
$k = 0,1,\dots, C-1$, and $m>m_{1}$, then (\ref{eq24.01.19.4}) can be  decomposed as (\ref{eq24.01.19.6}) by utilizing (\ref{eq24.01.08.1}), (\ref{eq24.01.11.3}), and the fact that $\phi_{m}(t)$ and $\phi_{m_{1}}(t)$ have finite support over the interval $[0, T]$, where 
\begin{align}
	\delta_{f}^{m,m_{1}}= (m-m_{1})\Delta f \label{eq24.01.25.10}
\end{align}
denotes the frequency difference between the two chirp subcarriers. Finally, after some algebraic manipulations, we obtain the final result in (\ref{eq24.01.19.11}), which exhibits a similar form to $A_{\phi_{m},\phi_{m}}(\tau, \nu)$ in (\ref{eq24.01.11.7}). Based on this, we visualize $A_{\phi_{m},\phi_{m_{1}}}(\tau, \nu)$ in Fig. \ref{fig.3-5} and propose the following proposition:
\begin{figure*}
	\begin{align}
		&A_{\phi_{m},\phi_{m_{1}}}(\tau, \nu) \overset{(\ref{eq24.01.11.2})}{=} \int_{-\infty}^{\infty} 
		e^{j 2 \pi\left(c_{2}m^{2}+\tilde{c}_{1}t^{2}+\frac{m}{T}t-\frac{q_{m}(t)}{\Delta t}t\right)}
		e^{-j 2 \pi\left[c_{2}m_{1}^{2}+\tilde{c}_{1}(t-\tau)^{2}+\frac{m_{1}}{T}(t-\tau)-\frac{q_{m_{1}}(t-\tau)}{\Delta t}(t-\tau)\right]}e^{-j2\pi\nu t}\text{d}t
		\label{eq24.01.19.4} \\
		&\overset{(\ref{eq24.01.20.1})}{=}
		e^{j 2 \pi \left[c_{2}(m^{2}-m_{1}^{2})-\tilde{c}_{1}\tau^{2} + \frac{m\tau}{T}
			\right]}
		\Bigg[
		\sum_{i=0}^{C-\tilde{q}_{m}(\tau)}
		\Bigg(\underbrace{e^{-j 2 \pi\frac{i}{\Delta t}\tau} 
		\int_{t_{m_{1},i}+\tau}^{t_{m,i+\tilde{q}_{m}(\tau)}}
		e^{j 2 \pi
			\left(\mathcal{E}^{\tau, \nu}_{m,\tilde{c}_{1}}+\delta_{f}^{m,m_{1}}
			\right)t} \text{d}t}_{\tilde{\mathcal{S}_{a}^{i}}}\notag\\
		&\qquad \quad  +
		\underbrace{e^{-j 2 \pi\frac{i}{\Delta t}\tau} \int_{t_{m,i+\tilde{q}_{m}(\tau)}}^{t_{m_{1},i+1}+\tau} 
		e^{j 2 \pi
			\left(\tilde{\mathcal{E}}^{\tau, \nu}_{m,\tilde{c}_{1}}+\delta_{f}^{m,m_{1}} 
			\right)t} \text{d}t}_{\tilde{\mathcal{S}}_{b}^{i}}  
		\Bigg)  
		+
		\underbrace{e^{-j 2 \pi\frac{C-\tilde{q}_{m}(\tau)+1}{\Delta t}\tau}  
		\int_{t_{m_{1},C-\tilde{q}_{m}(\tau)+1}+\tau}^{T} 
		e^{j 2 \pi
			\left(\mathcal{E}^{\tau, \nu}_{m,\tilde{c}_{1}}+\delta_{f}^{m,m_{1}} 
			\right)t} \text{d}t }_{\tilde{\mathcal{S}}_{a}^{C-\tilde{q}_{m}(\tau)+1}}\Bigg] \label{eq24.01.19.6} \\
		&=
		e^{j 2 \pi \left[c_{2}(m^{2}-m_{1}^{2})-\tilde{c}_{1}\tau^{2} + \frac{m\tau}{T}
			\right]}
		\Bigg[
		\sum_{i=0}^{C-\tilde{q}_{m}(\tau)}	
		e^{-j 2 \pi\frac{i}{\Delta t}\tau} 
		\Bigg( 
		\mathcal{I}\left(\left(\mathcal{E}^{\tau, \nu}_{m,\tilde{c}_{1}}+\delta_{f}^{m,m_{1}}\right), \left(t_{m_{1},i}+\tau\right), \left(t_{m,i+\tilde{q}_{m}(\tau)}\right)\right) \notag \\
		&
		+ \mathcal{I}\left(\left(\tilde{\mathcal{E}}^{\tau, \nu}_{m,\tilde{c}_{1}}+\delta_{f}^{m,m_{1}}\right), \left(t_{m,i+\tilde{q}_{m}(\tau)}\right), \left(t_{m_{1},i+1}+\tau\right)\right)
		\Bigg)
		+ e^{-j 2 \pi\frac{C-\tilde{q}_{m}(\tau)+1}{\Delta t}\tau
		}
		\mathcal{I}\left(\left(\mathcal{E}^{\tau, \nu}_{m,\tilde{c}_{1}}+\delta_{f}^{m,m_{1}}\right), \left(t_{m_{1},C-\tilde{q}_{m}(\tau)+1}+\tau\right), T\right)	\Bigg]. 
		\label{eq24.01.19.11}
	\end{align}
	\hrulefill
\end{figure*}

\begin{proposition} 
	The CAF between two AFDM chirp subcarriers establishes the same local and global properties as the AAF of the AFDM chirp subcarrier but includes an additional shift of $\delta_{f}^{m,m_{1}}$ along the Doppler dimension.
\end{proposition} 
\begin{proof}
	The conclusion from Sec. \ref{sec3-2}, stating that satisfying the FA and RPDA conditions boosts a pulse in the AF, also applies to any two distinct chirp subcarriers. This holds because the time-frequency distributions of all chirp subcarriers follow the same subchirp splicing pattern, as shown in Fig. \ref{fig.3-4}. Furthermore, we can observe that the only difference between the instant frequency part of  $\tilde{\mathcal{S}}_{a}^{i}$ ($\tilde{\mathcal{S}}_{b}^{i}$) in (\ref{eq24.01.19.6}) and that of $\mathcal{S}_{a}^{i}$ ($\mathcal{S}_{b}^{i}$) in (\ref{eq24.01.12.2}) is an additional term of $\delta_{f}^{m,m_{1}}$ in (\ref{eq24.01.19.6}), which indicates that the Doppler shifts that corresponds to case Fig. \ref{fig.3-3}(a) (as visuallized in Fig. \ref{fig.3-4}) and case in Fig. \ref{fig.3-3}(c) to satisfy the FA condition in $	A_{\phi_{m},\phi_{m_{1}}}(\tau, \nu)$ are
	\begin{equation}
		\nu=2\tilde{c}_{1}\tau 
		-\frac{\tilde{q}_{m}(\tau)-1}{\Delta t} +\delta_{f}^{m,m_{1}}  
		\label{eq24.01.29.1}
	\end{equation}
	and
	\begin{equation}
		\nu=2\tilde{c}_{1}\tau 
		-\frac{\tilde{q}_{m}(\tau)}{\Delta t}+\delta_{f}^{m,m_{1}},  
		\label{eq24.01.29.2}
	\end{equation}
	respectively. This results in a additional shift of $\delta_{f}^{m,m_{1}}$ of the locations of the pulses in $A_{\phi_{m},\phi_{m_{1}}}(\tau, \nu)$ on the basis of those in $A_{\phi_{m},\phi_{m}}(\tau, \nu)$, as indicated in Fig. \ref{fig.3-5}.  This completes the proof.
\end{proof}

Furthermore, from (\ref{eq24.01.19.11}), we observe that parameter $c_{2}$ affects the value of $A_{\phi_{m},\phi_{m_{1}}}(\tau, \nu)$ but not its magnitude. Building on these insights, we now extend the AF analysis of AFDM chirp subcarriers to AFDM frames.

\section{Ambiguity Functions of AFDM Frames}
\label{sec5}
In Sec. \ref{sec3}, we analyzed the AAF of a single chirp subcarrier, corresponding to the case where only a pilot symbol is transmitted in $\mathbf{x}$. In this section, we extend our study to the aperiodic AFs of various typical AFDM frames with different pilot-data structures. Specifically, we first examine the AAF of AFDM frames that only contain random data symbols and then investigate the CAF of AFDM frames that include both pilot and random data symbols.

\subsection{AAF of the Random-Data-Only AFDM Frame}
The AAF of $s(t)$, which is associated with the MF operation at the MSR as discussed in Sec. \ref{sec2-3}, can be obtained as
\begin{align}
	&A_{s,s}(\tau, \nu)= \int_{-\infty}^{\infty}
	s(t)
	s^{*}(t-\tau)e^{-j2\pi\nu t}\text{d}t \label{eq24.01.21.2} \\
	&=\int_{-\infty}^{\infty}
	\left(\sum_{m=0}^{N-1} x[m] \phi_{m}(t) \right) \notag \\
	&\qquad  \qquad \qquad  \times
	\left(\sum_{m_{1}=0}^{N-1} x^{*}[m_{1}] \phi_{m_{1}}^{*}(t-\tau) \right)
	e^{-j2\pi\nu t}\text{d}t \notag\\
	&=\sum_{m=0}^{N-1} 
	\sum_{m_{1}=0}^{N-1}
	x[m] x^{*}[m_{1}]
	\int_{-\infty}^{\infty}
	\phi_{m}(t) 
	\phi^{*}_{m_{1}}(t-\tau) 
	e^{-j2\pi\nu t}\text{d}t \notag\\
	&=\sum_{m=0}^{N-1} 
	\sum_{m_{1}=0}^{N-1} 
	x[m]x^{*}[m_{1}] A_{\phi_{m},\phi_{m_{1}}}(\tau, \nu),
	\label{eq24.01.21.1}
\end{align}
which is the superposition of the scaled AFs between all chirp subcarrier pairs. Owing to the random nature of the DAFT-domain data symbols in $\mathbf{x}$, $A_{s,s}(\tau, \nu)$ behaves as a random function. We hence study $A_{s,s}(\tau, \nu)$ statistically by considering $\mathbb{E}\left(\left|A_{s,s}(\tau, \nu)\right|^{2}\right) $ with (\ref{eq24.01.22.2}), where the expectation is with respect to $\mathbf{x}$. Considering that a unit power constellation is adopted and data symbols are independent of each other, i.e.,
\begin{align}
	\mathbb{E}\left(x[m]x^{*}[m]\right) =1, \quad \mathbb{E}\left(x[m]x^{*}[\tilde{m}]\right) = 0, \ \forall m\neq\tilde{m},
	\label{eq24.01.21.6}
\end{align}
we get (\ref{eq24.01.24.1}), which shows that  $\mathbb{E}\left(\left|A_{s,s}(\tau, \nu)\right|^{2}\right)$ is affected by the $\mathbb{E}\left(\left|x[m]\right|^{4}\right)$, $\mathbb{E}\left(x[m]^{2}\right)$, and $\mathbb{E}\left((x^{*}[m_{1}])^{2}\right)$ of the adopted constellation.
\begin{figure*}
	\begin{align}
		&\mathbb{E}\left(\left|A_{s,s}(\tau, \nu)\right|^{2}\right)  =\mathbb{E}\left(A_{s,s}(\tau, \nu)A_{s,s}^{*}(\tau, \nu)\right) \notag \\
		&  =\mathbb{E}\left(\left(\sum_{m=0}^{N-1}
		\sum_{m_{1}=0}^{N-1} 
		x[m]x^{*}[m_{1}] A_{\phi_{m},\phi_{m_{1}}}(\tau, \nu)\right) \left(\sum_{m_{2}=0}^{N-1} 
		\sum_{m_{3}=0}^{N-1} 
		x^{*}[m_{2}]x[m_{3}] A_{\phi_{m_{2}},\phi_{m_{3}}}^{*}(\tau, \nu)\right)\right) \notag \\
		& 
		=\sum_{m=0}^{N-1} 
		\sum_{m_{1}=0}^{N-1} 
		\sum_{m_{2}=0}^{N-1} 
		\sum_{m_{3}=0}^{N-1} 
		\mathbb{E}\left(x[m]x^{*}[m_{1}]
		x^{*}[m_{2}]x[m_{3}]\right)
		A_{\phi_{m},\phi_{m_{1}}}(\tau, \nu)
		A_{\phi_{m_{2}},\phi_{m_{3}}}^{*}(\tau, \nu)\label{eq24.01.22.2}\\
		&
		\overset{(\ref{eq24.01.21.6})}{=}
		\mathbb{E}\left(\left|x[m]\right|^{4}\right)
		\sum_{m=0}^{N-1}  
		\left|
		A_{\phi_{m},\phi_{m}}(\tau, \nu)\right|^{2}
		+
		\sum_{m=0}^{N-1}\sum_{m_{2}=0,m_{2}\neq m}^{N-1}
		A_{\phi_{m},\phi_{m}}(\tau, \nu)
		A_{\phi_{m_{2}},\phi_{m_{2}}}^{*}(\tau, \nu) \notag \\
		&\quad +\sum_{m=0}^{N-1}\sum_{m_{1}=0,m_{1}\neq m}^{N-1}
		A_{\phi_{m},\phi_{m_{1}}}(\tau, \nu)
		A_{\phi_{m},\phi_{m_{1}}}^{*}(\tau, \nu)
		+
		\sum_{m=0}^{N-1}\sum_{m_{1}=0,m_{1}\neq m}^{N-1}
		\mathbb{E}\left(x[m]^{2}\right)
		\mathbb{E}\left((x^{*}[m_{1}])^{2}\right)
		A_{\phi_{m},\phi_{m_{1}}}(\tau, \nu)
		A_{\phi_{m_{1}},\phi_{m}}^{*}(\tau, \nu).
		\label{eq24.01.24.1}
	\end{align}
	\hrulefill
\end{figure*}

\begin{figure*}[htbp]
	\centering
	\includegraphics[width=1\textwidth,height=0.285\textwidth]{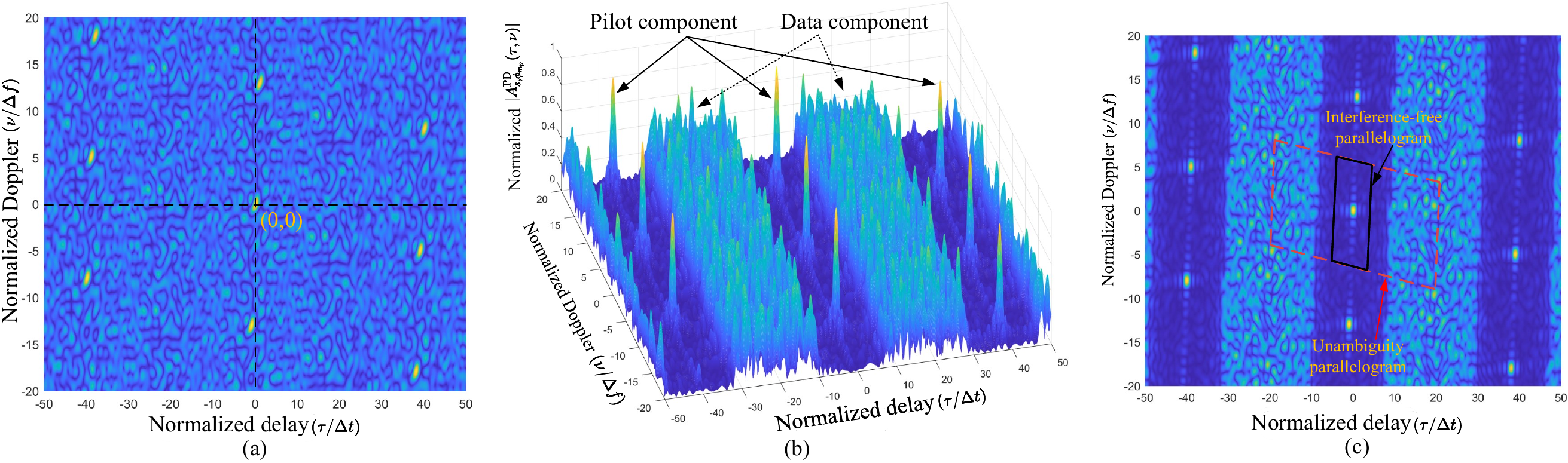}
	\caption{(a) Planform of the CAF between an SP AFDM frame and a pilot chirp subcarrier; (b) CAF between an EP AFDM
		frame and a pilot chirp subcarrier; (c) Planform of (b) ($N=512$, $C=13$, $c_{2} =\sqrt{2}$, $m_{p}=0$, $Q=20$).}
	\label{fig.4-2}
\end{figure*}

\begin{figure}[tbp]
	\centering
	\includegraphics[width=0.48\textwidth,height=0.17\textwidth]{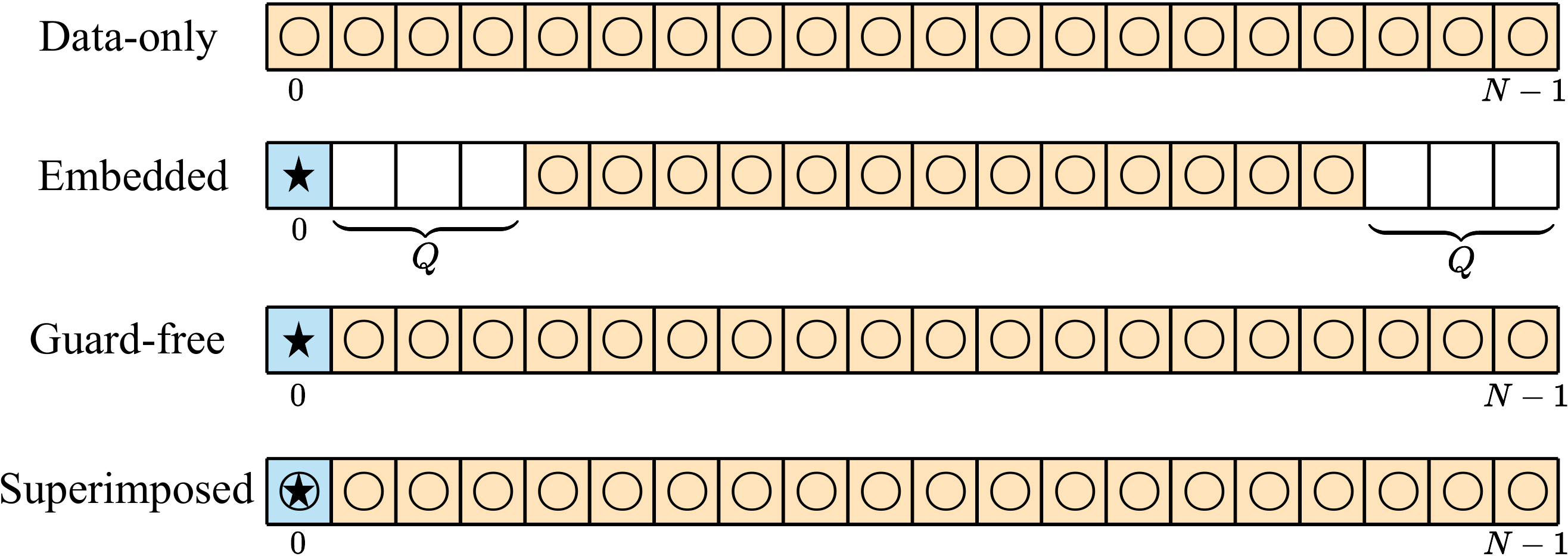}
	\caption{AFDM frames with different pilot-data structures, where each slot represents a chirp subcarrier (“$\bigcirc$”: Data; “$\star$”: Pilot;  blank: Guard).}
	\label{fig.4-1}
\end{figure}

\subsection{CAF of the AFDM Frames with Pilot and Random Data}
Next, we analyze AFDM frames that include pilot symbols.
Let $x_{p}$ and $m_{p}$ represent the pilot and its index in $\mathbf{x}$, respectively, $\mathcal{G}$ and $\mathcal{D}$ represent the index sets of guard symbols and data symbols, respectively, where $\mathcal{G} \cap  \mathcal{D}=\varnothing$, $m_{p}\notin \mathcal{G}$, and $Q$ represent the number of one-sided guard symbols, as illustrated in Fig. \ref{fig.4-1}. Particularly,  $\mathcal{D}=\{m|m=0,1,\dots,N-1\}$ and $\mathcal{D}=\{m|m=0,1,\dots,N-1, \ m\neq m_{p}\}$ correspond to the superimposed pilot (SP) structure in \cite{bb25.01.25.1} and the guard-interval-free pilot structure in \cite{bb25.01.25.2}, respectively, while $m_{p} \notin \mathcal{D}, \mathcal{D} \neq \varnothing  ,\mathcal{G} \neq \varnothing$ corresponds to the embedded pilot (EP) structure in \cite{bb24.08.27.2, 23.10.18.1}. Then, for the BSR that only knows the pilot part of $s(t)$, performing MF on $r_{\text{BSR}}(t)$ with the pilot chirp subcarrier is feasible. This corresponds to the CAF between $s(t)$ and $\phi_{m_{p}}(t)$ as
\begin{align}
	&A_{s,\phi_{m_{p}}}^{\text{PD}}(\tau, \nu)= \int_{-\infty}^{\infty}
	s(t)
	\phi^{*}_{m_{p}}(t-\tau)e^{-j2\pi\nu t}\text{d}t \notag \\
	&\qquad \ \ =\int_{-\infty}^{\infty}
	\left(\sum_{m=0}^{N-1} x[m] \phi_{m}(t) \right) \phi^{*}_{m_{p}}(t-\tau)e^{-j2\pi\nu t}\text{d}t \notag \\
	&\qquad \ \ =
	\underbrace{\left|x_{p}\right|^{2}A_{\phi_{m_{p}},\phi_{m_{p}}}(\tau, \nu)}_{\text{Pilot component}}
	+
	\underbrace{x_{p}^{*}\sum_{m\in\mathcal{D}}
	x[m] A_{\phi_{m},\phi_{m_{p}}}(\tau, \nu)}_{\text{Data component}},
	\label{eq24.01.21.3}
\end{align}
which can be divided into two parts: the pilot component and the data component.

To proceed, we define the pilot-power-to-data-power ratio (PDR) as $\frac{\left|x_{p}\right|^{2}}{\mathbb{E}\left(\left|x[m]\right|^{2}\right)}$. Fig. \ref{fig.4-2}(a) shows the CAF of the SP case with a practical PDR of 10 dB. We can observe that the pulses in the CAF of the AFDM frame with SP are not as distinguishable as those in the AAF of the AFDM chirp subcarrier, as shown in Fig. \ref{fig.3-1}(a). This occurs because the pulses are essentially generated by the AAF of the pilot chirp subcarrier, i.e., the pilot component in (\ref{eq24.01.21.3}), while the additional data component in the SP CAF introduces blurring, causing the pulses to become less distinct. The data component in the SP CAF has non-trivial energy and is distributed over the entire DD plane of the AF despite the relatively small energy of the data symbols given that they are all matched filtered with the pilot symbol. This significantly degrades the sensing performance at the BSR.

In addition to increasing the pilot energy to mitigate the interference from the data component, an effective solution is to introduce guard symbols around the pilot symbol. Fig. \ref{fig.4-2}(b) shows an example of $A_{s,\phi_{m_{p}}}^{\text{PD}}(\tau, \nu)$ with an EP structure, where the ratio of the number of pilot and guard symbols to the total number of chirp subcarriers, denoted as $\rho$, is set to $45.5\%$. We can observe that guard bands are formed around the pilot component. This phenomenon can be explained using Proposition 2: the CAF contributions between the chirp subcarriers within $\mathcal{G}$ and the pilot chirp subcarrier vanish due to the deployment of guard symbols. Consequently, as shown in the view of  $A_{s,\phi_{m_{p}}}^{\text{PD}}(\tau, \nu)$ with the EP structure in Fig. \ref{fig.4-2}(c), a consecutive arrangement of guard symbols around the pilot creates guard bands along the same rotated Doppler dimension as the AAF of the AFDM chirp subcarrier. This enhances the identifiability of the pilot component in the EP CAF, thereby improving sensing accuracy. 

\begin{remark} 
	\textup{We can further define a smaller parallelogram that occupies half of the intersection region between the guard band and the unambiguity parallelogram, as indicated by the black solid line in Fig. \ref{fig.4-2}(c). When the delay and Doppler shifts of all targets in the sensing channel fall within this smaller parallelogram, the original pulse remains confined within it. Consequently, in the MF output at the BSR, the data component does not spread into this region. This ensures interference-free sensing, leading us to refer to this smaller parallelogram as the \textbf{\emph{interference-free parallelogram}}. This interference-free parallelogram is jointly determined by $c_{1}$ and $\rho$ for a given $N$. Note that the interference-free parallelogram is always smaller than the unambiguity parallelogram, implying a stricter requirement on the sensing channel to achieve interference-free sensing at the BSR compared to unambiguous sensing at the MSR.}
\end{remark}

\section{Simulation Results}
\label{sec6}
\begin{figure}[tbp]
	\centering
	\includegraphics[width=0.47\textwidth,height=0.4\textwidth]{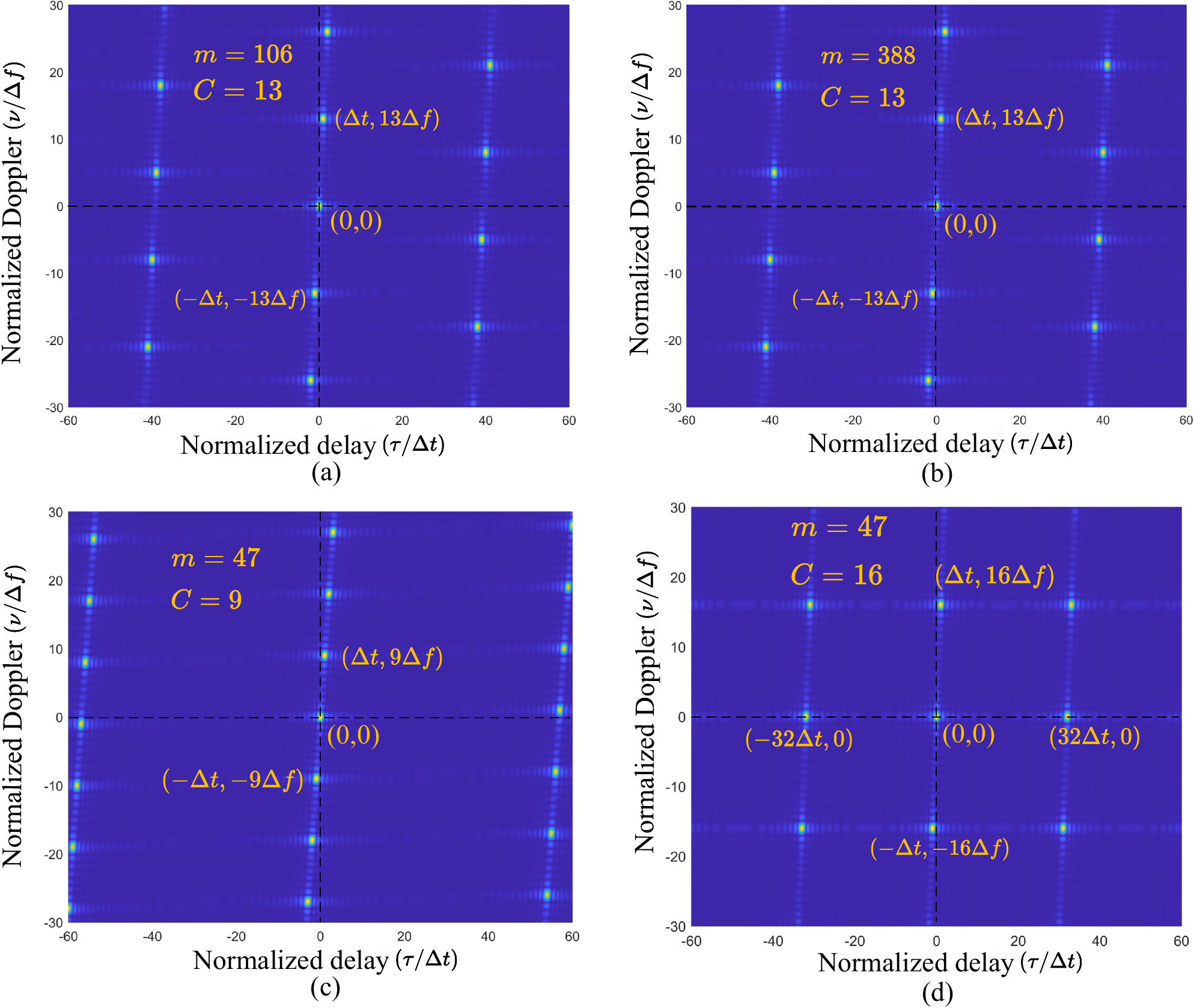}
	\caption{Planforms of the AAFs of different AFDM chirp subcarriers, $N=512$: (a) $C=13, m=106$; (b) $C=13, m=388$; (c) $C=9, m=47$; (d) $C=16, m=47$.}
	\label{fig.6-1}
\end{figure}

In this section, we present simulation results that validate our theoretical analysis. To begin with, Fig. \ref{fig.6-1}(a) and Fig. \ref{fig.6-1}(b) show the planforms of the AAFs of AFDM chirp subcarriers with  $m=106$ and $m=388$, respectively. We can observe that they establish the same local and global properties as the case in Fig. \ref{fig.3-1}(b) with $m=47$, proving the effectiveness of the analysis in Sec. \ref{sec3} to all AFDM chirp subcarriers. This is because all AFDM chirp subcarriers occupy the entire time and frequency resource of an AFDM frame and establish the same subchirp splicing structure with a common chirp slope. Moreover, Fig. \ref{fig.6-1}(c) and Fig. \ref{fig.6-1}(d) show the planforms of the AAFs of different AFDM chirp subcarriers with $C=9$ and $C=16$, respectively. We observe that the relative delay distance and the relative Doppler distance between two adjacent pulses along the rotated Doppler dimension are $\Delta t$ and $C\Delta f$, respectively, verifying Remark 1. In particular, we can observe that the pulses are distributed along the delay dimension for the case in Fig. \ref{fig.6-1}(d). This is because $T_{\text{SC}}=\frac{T}{C}=32 \Delta t$ satisfies the RPDA condition in (\ref{eq24.01.14.2}), leading to a no rotation along the delay dimension in the distribution of pulses in $A_{\phi_{m},\phi_{m}}(\tau, \nu)$, which aligns well with Remark 2.

\begin{figure}[tbp]
	\centering
	\includegraphics[width=0.47\textwidth,height=0.425\textwidth]{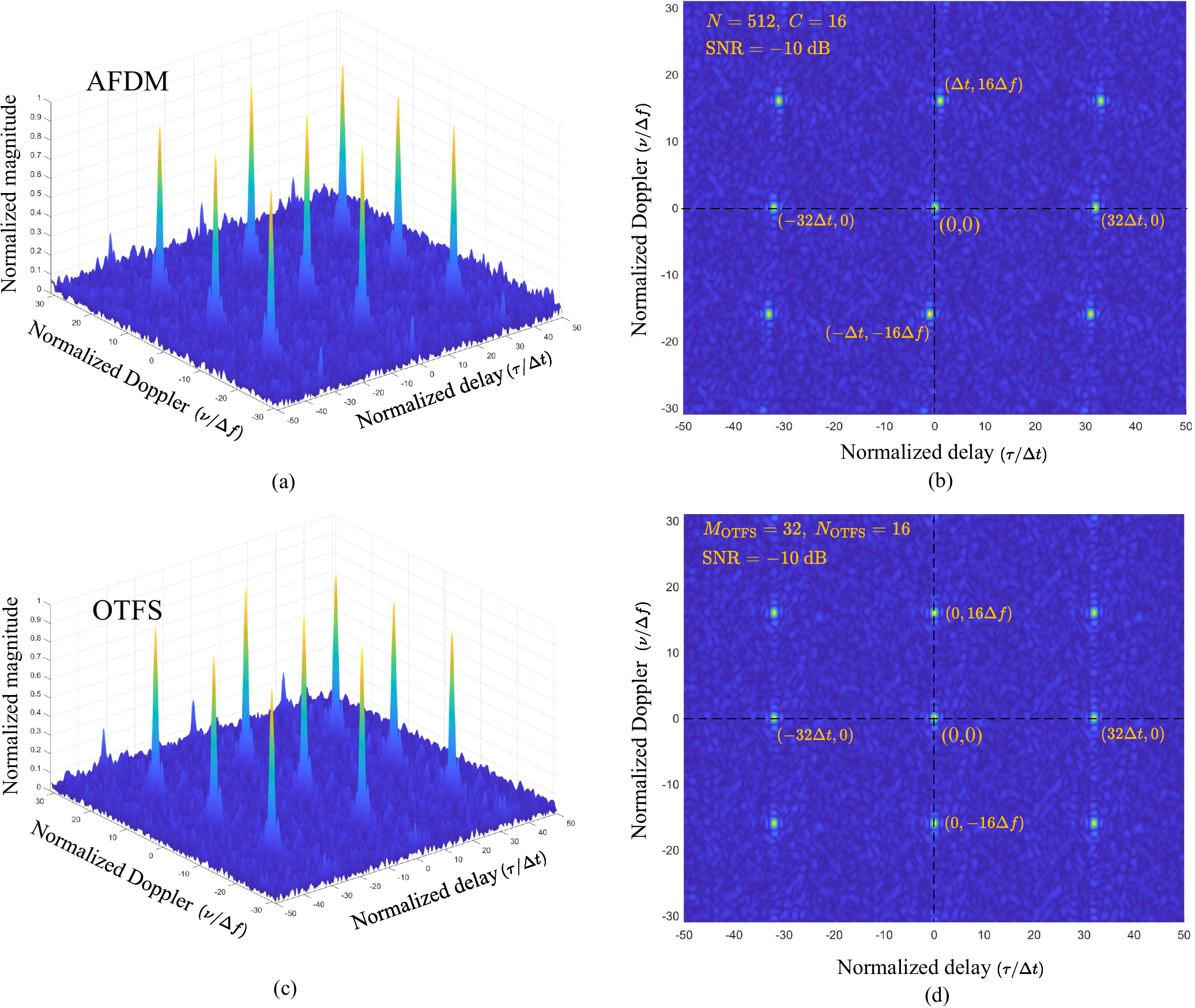}
	\caption{AAFs of AFDM chirp subcarriers and OTFS subcarriers generated by RRC filter with a roll-off factor of 0.25, rectangular pulse shaping, -10 dB SNR: (a) AAF of AFDM chirp subcarriers, $N=512$, $C=13$; (b) Planform of (a); (c) AAF of OTFS subcarriers \cite{bb24.03.15.3}, $M_{\text{OTFS}}=32$, $N_{\text{OTFS}}=16$; (d) Planform of (c).}
	\label{fig.6-1-2}
\end{figure}

Fig. \ref{fig.6-1-2} compares the AAFs of AFDM chirp subcarriers and OTFS subcarriers. The discrete-time signals are filtered with a root-raised-cosine (RRC) pulse with a roll-off factor of 0.25, combined with rectangular pulse shaping and a signal-to-noise ratio (SNR) of -10 dB \cite{bb24.9.15.1, bb24.03.15.3}. To ensure both waveforms occupy identical time-frequency resources, we set $N=512$, $M_{\text{OTFS}}=32$, and $N_{\text{OTFS}}=16$, so that $N=M_{\text{OTFS}}N_{\text{OTFS}}$, where $N_{\text{OTFS}}$ and $N_{\text{OTFS}}$ denote the numbers of Doppler and delay bins, respectively, in the DD block of an OTFS frame. Figs. \ref{fig.6-1-2}(a)(b) show that the AAF of RRC-filtered AFDM subcarrier retains the “spike-like” local property and “periodic-like” global property observed for the AAF of the analytic signal $\phi_{m}(t)$, confirming that the AF analysis developed in this paper remains valid for continuous-time AFDM signals generated via practical digital-to-analog filtering. Furthermore, the AAFs of AFDM chirp subcarriers and OTFS subcarriers exhibit strikingly similar local and global characteristics. This resemblance arises because both subcarriers fully occupy the same time-frequency grid and contain inherent periodic structures, the subchirp splicing in AFDM chirp subcarriers and the “pulse-train” features in the time- and frequency-domain representations of OTFS subcarriers \cite{bb25.01.21.1, bb24.03.15.3}.

We then examine the CAFs of different AFDM chirp subcarrier pairs. Fig. \ref{fig.6-2}(a) shows the CAF between the $m=38$ and $m_{1}=33$ chirp subcarriers, which establishes the same pulse distribution as in the case of Fig. \ref{fig.3-5} with $m=445$ and $m_{1}=440$. This is because they have the same frequency difference of $\delta_{f}^{m,m_{1}}=5\Delta f$. Moreover, Fig. \ref{fig.6-2}(b) shows the CAF between the $m=38$ and $m_{1}=33$ chirp subcarriers ($\delta_{f}^{m,m_{1}}=-3\Delta f$), where an additional shift of $-3\Delta f$ along the Doppler dimension compared to $A_{\phi_{m},\phi_{m}}(\tau, \nu)$ can be clearly observed.
These verify the effectiveness of the analysis in Sec. \ref{sec4} to all AFDM chirp subcarrier pairs.

\begin{figure}[tbp]
	\centering
	\includegraphics[width=0.48\textwidth,height=0.22\textwidth]{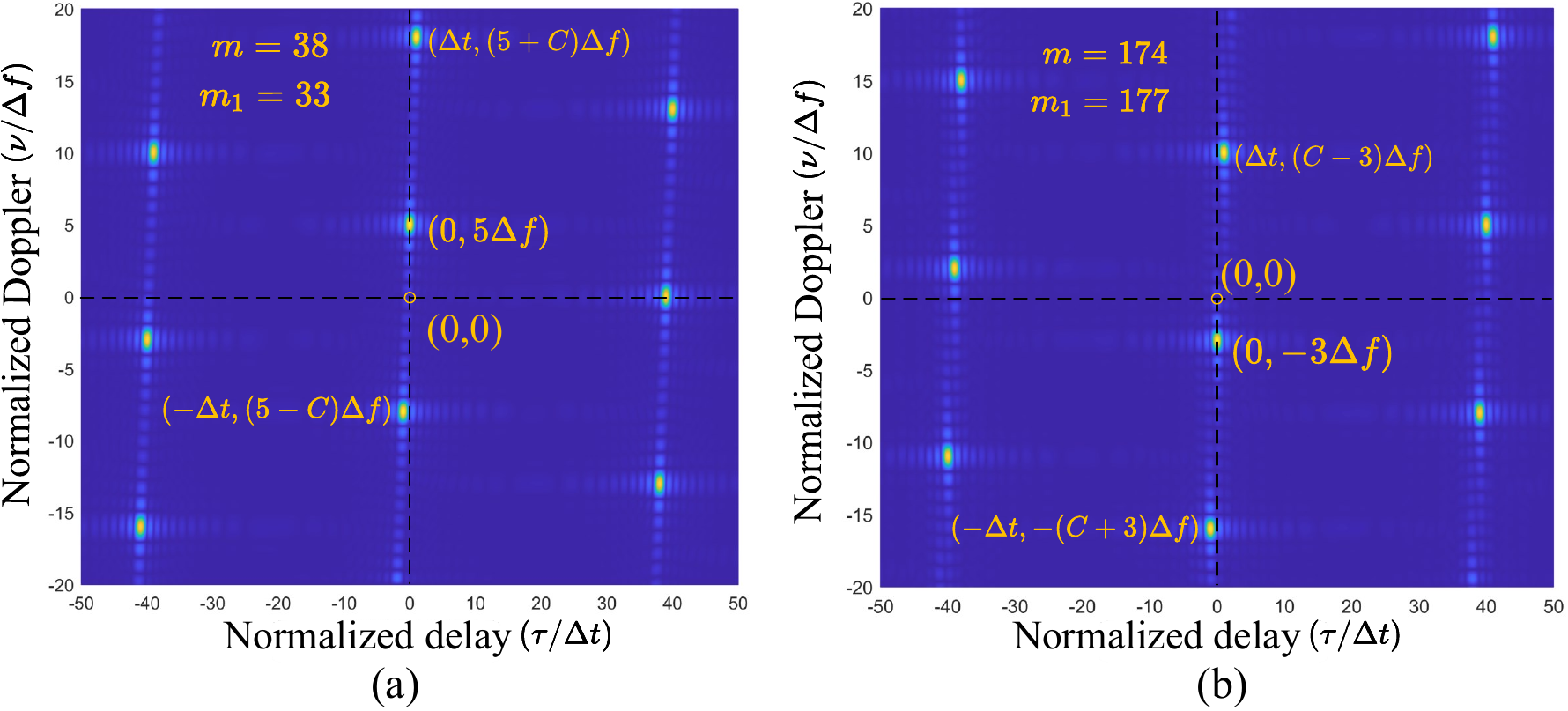}
	\vspace{-0.5em}
	\caption{Planforms of the CAFs of different chirp subcarriers,  $N=512$, and $c_{2} =0$: (a) $C=13$, $m=38$, $m_{1}=33$   ($\delta_{f}^{m,m_{1}}=5\Delta f$); (b) $C=13$, $m=174$, $m_{1}=177$   ($\delta_{f}^{m,m_{1}}=-3\Delta f$).}
	\label{fig.6-2}
\end{figure}

\begin{figure}[tbp]
	\centering
	\includegraphics[width=0.48\textwidth,height=0.219\textwidth]{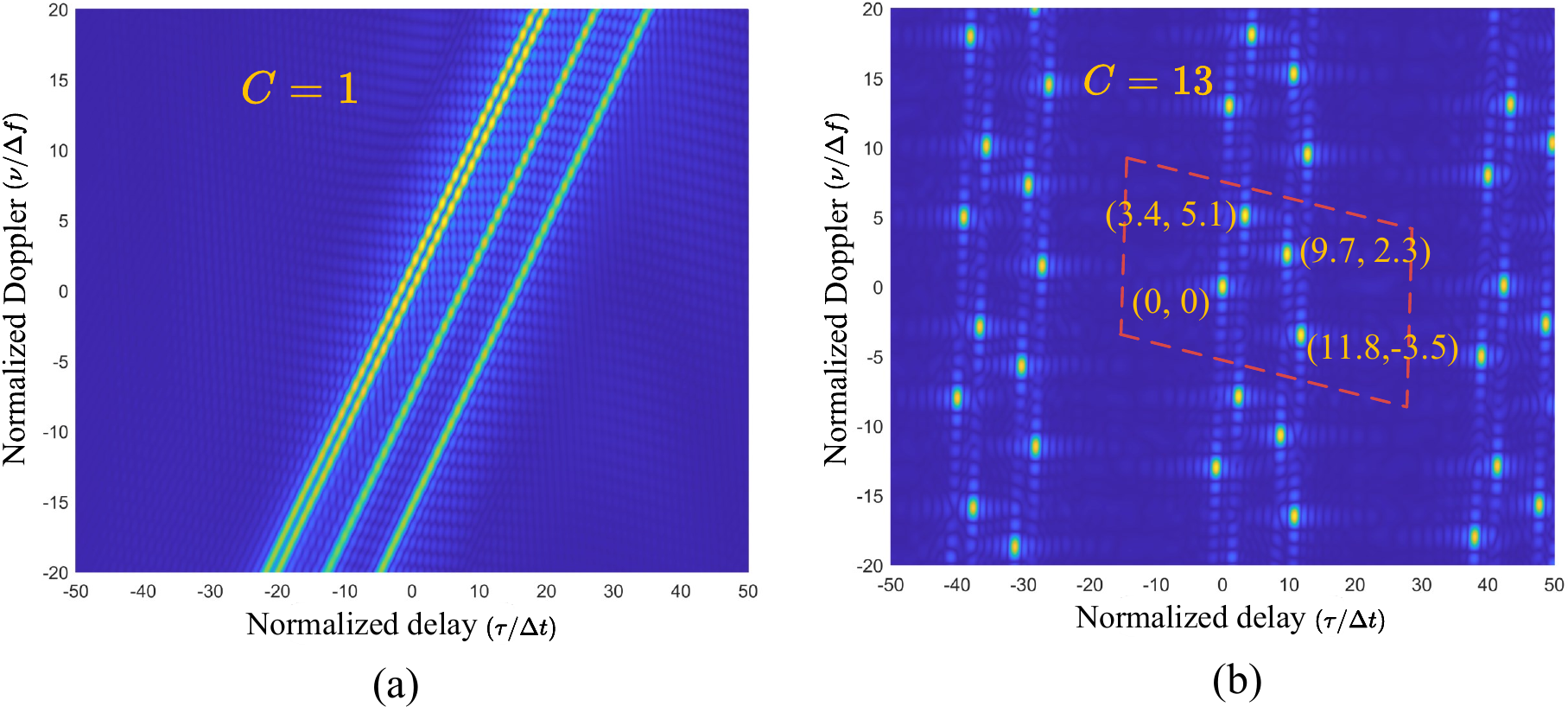}
	\vspace{-0.5em}
	\caption{Planforms of the output of MF with a chirp subcarrier is transmitted, four targets with delay shifts of $[0, 3.4, 9.7, 11.8]\Delta t$ and Doppler shifts of $[0, 5.1, 2.3, -3.5]\Delta f$: (a) $C=1$; (b) $C=13$.}
	\label{fig.6-3}
\end{figure}

Next, we demonstrate the influence of the parameter $c_{1}$ on the sensing ability of AFDM chirp subcarriers with Fig. \ref{fig.6-3}. We adopt a carrier frequency $f_{c}=72$ GHz, with $N=512$, and $\Delta f= 12$ kHz, resulting in a signal bandwidth of $B=6.1$ MHz and a duration of $T=83.3 \ \mu$s. Assume a sensing scenario involving four targets with round-trip characteristics (range in meters, velocity in km/h) of (0, 0), (331.8, 918.0), (946.6, 414.0), and (1,151.5, $-630$), corresponding to delay-Doppler shifts of $(0,\ 0)$, $(3.4\Delta t,\  5.1\Delta f)$, $(9.7\Delta t,\ 2.3 \Delta f)$, and $(11.8\Delta t,\ -3.5\Delta f)$, respectively, and a pilot chirp subcarrier is transmitted. Fig. \ref{fig.6-3}(a) shows the MF output at the sensing receiver for $C=1$, corresponding to an approximate maximum unambiguous delay shift of $T_{\text{sc}}=512 \Delta t$ and Doppler shift of $C\Delta f=\Delta f$. In this scenario, all four targets can be preliminarily distinguished. However, the unambiguity region, represented by a narrow parallelogram, is insufficient to fully encompass the sensing channel, preventing accurate resolution of all targets' delay and Doppler shifts. In contrast, Fig. \ref{fig.6-3}(b) shows the case with $C=13$, corresponding to an approximate maximum unambiguous delay shift of $39.4 \Delta t$ and Doppler shift of $13\Delta f$. In this case, all four targets fall within the unambiguity parallelogram, allowing their delay and Doppler shifts to be fully determined. This observation aligns well with Proposition \ref{prop1}.

We proceed to investigate the AAF of AFDM frames with random data. All numerical results are obtained by averaging over 10,000 Monte Carlo simulations. Fig. \ref{fig.6-4}(a) shows the magnitude of $\mathbb{E}\left(\left|A_{s,s}^{\text{Data}}(\tau, \nu)\right|^{2}\right)$. Notably, we observe a single energy-concentrated pulse at the origin, which differs from the AAF of AFDM chirp subcarriers. This can be explained by considering the summation in (\ref{eq24.01.21.1}). The pulses that do not localize at the origin - originating from the AAFs of individual chirp subcarriers and the CAFs between different chirp subcarriers - are averaged out. In contrast, the pulses at the origin from the AAFs of all chirp subcarriers share the same value ($A_{\phi_{m},\phi_{m}}(0, 0)=T$, $m=0,1,\dots,N-1$), resulting in a single peak in $\left|A_{s,s}^{\text{Data}}(\tau, \nu)\right|^{2}$. This finding confirms the feasibility of AFDM-ISAC frames with i.i.d data for radar sensing. Additionally, Fig. \ref{fig.6-4}(b) and Fig. \ref{fig.6-4}(c) show the zero-Doppler cut and zero-delay cut of $\mathbb{E}\left(\left|A_{s,s}^{\text{Data}}(\tau, \nu)\right|^{2}\right)$, respectively. The theoretical results align perfectly with their numerical counterparts, further validating our analysis. 

\begin{figure*}[tbp]
	\centering
	\includegraphics[width=1\textwidth,height=0.293\textwidth]{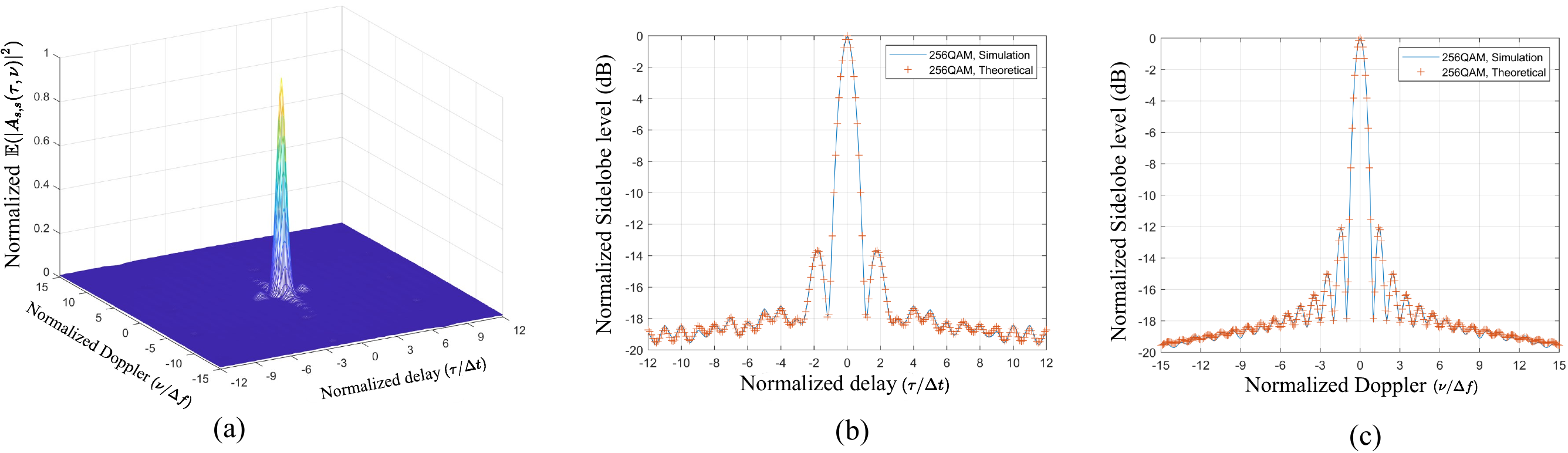}
	\caption{AAF of an AFDM frame with random data, $N=64$, $265$-QAM, $C=7$, and $c_{2}=0$: (a): Simulation of  $\mathbb{E}\left(\left|A_{s,s}^{\text{Data}}(\tau, \nu)\right|^{2}\right)$; (b) Zero-Doppler cut of $\mathbb{E}\left(\left|A_{s,s}^{\text{Data}}(\tau, \nu)\right|^{2}\right)$; (c) Zero-delay cut of $\mathbb{E}\left(\left|A_{s,s}^{\text{Data}}(\tau, \nu)\right|^{2}\right)$.}
	\label{fig.6-4}
\end{figure*}

\section{Conclusion}
\label{sec7}
In this paper, we provided a comprehensive analysis of the ambiguity functions of continuous-time AFDM signals. First, we derived the AAF of AFDM chirp subcarriers and characterized their “spike-like” local property and “periodic-like” global property along the rotated delay and Doppler dimensions. Building on this, we defined an unambiguity parallelogram and demonstrated that unambiguous sensing can be achieved by appropriately adjusting the parameter $c_{1}$. Next, we extended our analysis to the CAF between different AFDM chirp subcarriers, revealing that the CAF introduces an additional shift along the Doppler dimension, which is equal to the frequency difference between the two chirp subcarriers. Finally, we investigated the AFs of various AFDM frame structures with different pilot-data allocations, highlighting that the insertion of guard symbols in AFDM facilitates interference-free bistatic sensing. For future work, we plan to extend our AF analysis of AFDM signals to incorporate different pulse shaping designs.

{\appendices
	\section{Proof of  Corollary 1}
	\label{APP1}
	Without loss of generality, we calculate the area of the parallelogram formed by pulses (b, c, e, f) in Fig. \ref{fig.3-1}(b). Firstly, the base of the parallelogram is line (b, e), whose length can be calculated according to the coordinates of pulse “b”  and pulse “c” discussed in Remark 1 as
	\begin{equation}
		l_{b,e} = \sqrt{\Delta t^{2}+(2\tilde{c}_{1}\Delta t)^{2}}.
		\label{eq24.01.19.1}
	\end{equation}
	Then we construct a height of the parallelogram with a solid line that starts from point “c” and is perpendicular to line (b, e) at point “j”. Additionally, we further draw an auxiliary line along the time-delay dimension with a dash-dot line, which starts from point “c” and intersects line (b, e) at point “s”. According to Remark 2, the length of line (c, s) is $l_{c,s}=d_{\text{delay}}=T_{\text{SC}}$ and $\tan_{\angle j,s,c}=2\tilde{c}_{1}$, where $\angle j,s,c$ represents the angle (j, s, c). Subsequently, by applying the Pythagorean theorem on the triangle (c, j, s), we have
	\begin{equation}
		l_{c,j} = \sqrt{\frac{1}{\Delta t^{2}+(2\tilde{c}_{1}\Delta t)^{2}}}.
		\label{eq24.01.19.2}
	\end{equation}
Finally, the area of parallelogram (b, c, e, f) can be calculated as $l_{b,e}l_{c,j}=1$. This completes the proof of Corollary 1.
}

\vfill

\begin{thebibliography}{99} 
\bibliographystyle{IEEEtran}

\bibitem{bb25.2.4.1}
M. Chafii, L. Bariah, S. Muhaidat, and M. Debbah, “Twelve scientific challenges for 6G: Rethinking the foundations of communications theory,” \textit{IEEE Commun. Surveys Tuts.}, vol. 25, no. 2, pp. 868–904, 2023.

\bibitem{bb22.10.24.2}
T. Wang, J. G. Proakis, E. Masry, and J. R. Zeidler, “Performance degradation of OFDM systems due to doppler spreading,” \textit{IEEE Trans. Wireless Commun.}, vol. 5, no. 6, pp. 1422-1432, 2006.

\bibitem{bb25.01.08.102}
ITU-R WP5D, “Draft New Recommendation ITU-R M. [IMT. Framework for 2030 and Beyond],” 2023.

\bibitem{bb25.01.08.1}
F. Liu, Y. Xiong, S. Lu \textit{et al.}
“Integrated sensing and communications: Toward dual-functional wireless networks for 6G and beyond,” \textit{IEEE J. Sel. Areas Commun.}, vol. 40, no. 6, pp. 1728–1767, 2022.


\bibitem{bb2}
R. Hadani, S. Rakib, M. Tsatsanis \textit{et al.}, “Orthogonal time frequency space modulation,” \textit{IEEE Wireless Commun. Netw. Conf. (WCNC)}, pp. 1-6, 2017.

\bibitem{bb24.08.21.2}
Z. Wei, W. Yuan, S. Li \textit{et al.}, “Orthogonal time-frequency space modulation: A promising next-generation waveform," \textit{IEEE Wireless Commun.}, vol. 28, no. 4, pp. 136-144, 2021.


\bibitem{bb25.01.21.1}
S. K. Mohammed, R. Hadani, A. Chockalingam, and R. Calderbank, “OTFS-a mathematical foundation for communication and radar sensing in the delay-doppler domain," \textit{IEEE BITS Inf. Theory Mag.}, vol. 2, no. 2, pp. 36-55, 2022.

\bibitem{bb24.03.15.2}
H. Lin and J. Yuan, “Orthogonal delay-doppler division multiplexing modulation,” \textit{IEEE Trans. Wireless Commun.}, vol. 21, no. 12, pp. 11 024-11 037, Dec. 2022.



\bibitem{bb25.02.11.2}
Z. Xiao, X. Liu, Y. Zeng \textit{et al.}, “Rethinking waveform for 6G: Harnessing delay-doppler alignment modulation," \textit{IEEE Commun. Mag.}, early access, 2024.


\bibitem{bb25.01.21.2}
S. K. Mohammed, R. Hadani, A. Chockalingam, and R. Calderbank, “OTFS-predictability in the delay-doppler domain and its value to communication and radar sensing," \textit{IEEE BITS Inf. Theory Mag.}, vol. 3, no. 2, pp. 7-31, 2023.



\bibitem{bb24.03.15.311}
L. Gaudio, M. Kobayashi, G. Caire, and G. Colavolpe, “On the effectiveness of OTFS for joint radar parameter estimation and communication,” \textit{IEEE Trans. Wireless Commun.}, vol. 19, no. 9, pp. 5951–5965, 2020.




\bibitem{bb24.03.15.3}
S. Li, P. Jung, W. Yuan  \emph{et al.}, “Fundamentals of delay-Doppler communications: Practical implementation and extensions to OTFS," \textit{arXiv preprint	arXiv:2403.14192}, 2024.

\bibitem{bb25.02.11.1}
H. Lu and Y. Zeng, “Delay-doppler alignment modulation for spatially sparse massive MIMO communication," \textit{IEEE Trans. Wireless Commun.}, vol. 23, no. 6, pp. 6000-6014, 2024.

\bibitem{bb24.08.27.1} X. Ouyang and J. Zhao, “Orthogonal chirp division multiplexing,” \textit{IEEE Trans. Commun.}, vol. 64, no. 9, pp. 3946–3957, 2016. 

\bibitem{bb24.08.27.2}
A. Bemani, N. Ksairi, and M. Kountouris, “Affine frequency division multiplexing for next-generation wireless communications,” \textit{IEEE Trans. Wireless Commun.}, vol. 22, no. 11, pp. 8214-8229, 2023.


\bibitem{bb8}
T. Erseghe, N. Laurenti, and V. Cellini, “A multicarrier architecture based upon the affine fourier transform,” \textit{IEEE Trans. Commun.}, vol. 53, no. 5, pp. 853-862, May 2005.


\bibitem{bb25.01.08.2}
H. Yin, Y. Tang, A. Bemani \textit{et al.}, “Affine frequency division multiplexing: Extending OFDM for scenario-flexibility and resilience,” \textit{accepted by IEEE Wireless Commun.}, 2025.


\bibitem{bb24.9.08.100}
H. S. Rou, G. T. F. d. Abreu, J. Choi \textit{et al.}, “From orthogonal time–frequency space to affine frequency-division multiplexing: A comparative study of next-generation waveforms for integrated sensing and communications in doubly dispersive channels," \textit{IEEE Signal Process. Mag.}, vol. 41, no. 5, pp. 71-86,  2024.



\bibitem{23.10.18.1} 
H. Yin, X. Wei, Y. Tang, and K. Yang, “Diagonally reconstructed channel estimation for MIMO-AFDM with inter-doppler interference in doubly selective channels,”  \textit{IEEE Trans. Wireless Commun.}, vol. 23, no. 10, pp. 14066-14079, 2024.


\bibitem{bb25.01.25.1}
K. Zheng, M. Wen, T. Mao \textit{et al.}, “Channel estimation for AFDM with superimposed pilots," \textit{IEEE Trans. Veh. Technol.}, vol. 74, no. 2, 2025.

\bibitem{bb25.01.25.2} 
Y. Zhou, H. Yin, N. Zhou \textit{et al.}, “GI-free pilot-aided channel estimation for affine frequency division multiplexing systems,” \textit{arXiv preprint arXiv:2404.01088}, 2024.




\bibitem{bb24.9.15.1} 
H. Yin, Y. Tang, S. Li \textit{et al.}, “Evaluation and design criterion  for pulse-shaped AFDM,” \textit{IEEE Global Commun. Conf. (GLOBECOM)}, pp. 4944-4949, 2024.

\bibitem{bb24.9.08.3}
V. Savaux, “Special cases of DFT-based modulation and demodulation for affine frequency division multiplexing," \textit{IEEE Trans. Commun.}, vol. 72, no. 12, pp. 7627-7638, 2024.

\bibitem{bb24.9.08.2}
J. Du, Y. Tang, H. Yin  \textit{et al.}, “A simplified affine frequency division multiplexing system for high mobility communications,"  \textit{IEEE Wireless Commun. Netw. Conf. (WCNC)}, pp. 1-5, 2024.

\bibitem{bb23.3.6.1} H. Yin,  Y. Zhou, Y. Tang \textit{et al.}, “Cyclic delay-Doppler shift: A simple transmit diversity technique for ultra-reliable communications in doubly selective channels,” \textit{under review in IEEE Trans. Wireless Commun.} 2025.

\bibitem{bb25.01.09.4}
A. Bemani, N. Ksairi, and M. Kountouris, “Low complexity equalization for AFDM in doubly dispersive channels,” \textit{IEEE Int. Conf. Acoustics
	Speech Signal Process. (ICASSP)}, pp. 5273-5277, 2022.

\bibitem{bb25.01.25.3}
Z. Li, C. Zhang, G. Song \textit{et al.}, “Chirp parameter selection for affine frequency division multiplexing with MMSE equalization," \textit{IEEE Trans. Commun.}, early access, 2024.

\bibitem{bb24.9.08.1}
Q. Luo, P. Xiao, Z. Liu \textit{et al.}, “AFDM-SCMA: A promising waveform for massive connectivity over high mobility channels," \textit{IEEE Trans. Wireless Commun.}, vol. 23, no. 10, pp. 14421-14436, 2024.


\bibitem{bb24.9.08.10}
Y. Tao, M. Wen, Y. Ge \textit{et al.}, “DAFT-spread affine frequency division multiple access for downlink transmission," \textit{IEEE Global Commun. Conf. (GLOBECOM)}, pp. 2991-2996, 2024.



\bibitem{bb24.03.15.8}
J. Zhu, Q. Luo, G. Chen \textit{et al.}, “Design and performance analysis of index modulation empowered AFDM system," \textit{IEEE Wireless Commun. Lett.}, vol. 13, no. 3, pp. 686-690,  2024.


\bibitem{bb24.9.23.2} 
Y. Tao, M. Wen, Y. Ge \textit{et al.}, “Affine frequency division multiplexing with index modulation: Full diversity condition, performance analysis, and low-complexity detection,"  \textit{IEEE J. Sel. Areas Commun.}, vol. 43, no. 4, pp. 1041-1055, 2025.


\bibitem{bb24.03.15.10} 
G. Liu, T. Mao, Z. Xiao \textit{et al.}, “Pre-chirp-domain index modulation for full-diversity affine frequency division multiplexing towards 6G," \textit{IEEE Trans. Wireless Commun.}, early access, 2025.


\bibitem{bb24.03.15.1230} 
H. S. Rou, K. Yukiyoshi, T. Mikuriya \textit{et al.}, “AFDM chirp-permutation-index modulation with quantum-accelerated codebook design," \textit{Asilomar Conf. Signals Syst. Comput.}, pp. 817-821, 2024.




\bibitem{b9}  Q. Wang, A. Kakkavas, X. Gong and R. A. Stirling-Gallacher, “Towards integrated sensing and communications for 6G,” \textit{IEEE Int. Symp.  Joint Commun. \& Sensing (JC\&S)}, 2022.


\bibitem{bb24.9.08.103}
Y. Zhou, C. Zou, N. Zhou \textit{et al.}, “Affine frequency division multiplexing for communication and channel sounding: Requirements, challenges, and key technologies,” \textit{under review in IEEE Veh. Technol. Mag.}, 2025.


\bibitem{bb2022.11.11.2}
Y. Ni, P. Yuan, Q. Huang \textit{et al.}, “An integrated sensing and communications system based on affine frequency division multiplexing,"  \textit{IEEE Trans. Wireless Commun.}, vol. 24, no. 5, pp. 3763-3779, 2025.

\bibitem{bb25.1.21.10}
W. Benzine, A. Bemani, N. Ksairi, and D. Slock, “Affine frequency division multiplexing for compressed sensing of time-varying channels," \textit{IEEE Int. Workshop Signal Process. Advances Wireless Commun. (SPAWC)}, pp. 916-920, 2024.


\bibitem{bb2025.6.26.1}
	K. R. R. Ranasinghe, H. Seok Rou, G. Thadeu Freitas de Abreu \textit{et al.}, “Joint channel, data, and radar parameter estimation for AFDM systems in doubly-dispersive channels," \textit{IEEE Trans. Wireless Commun.}, vol. 24, no. 2, pp. 1602-1619, 2025.

\bibitem{bb24.03.15.5}
A. Bemani, N. Ksairi, and M. Kountouris, “Integrated sensing and communications with affine frequency division multiplexing,” \textit{IEEE Wireless Commun. Lett.}, vol. 13, no. 5, pp. 1255-1259, 2024.


\bibitem{bb24.9.08.104}
Y. Luo, Y. L. Guan, Y. Ge \textit{et al.}, “A novel angle-delay-Doppler estimation scheme for AFDM-ISAC system in mixed near-field and far-field scenarios," \textit{IEEE Internet Things J.}, early access, 2025.


\bibitem{bb25.01.09.1}
Z. Liao, F. Liu, S. Li \textit{et al.}, “Pulse shaping for random ISAC signals: The ambiguity function between symbols matters,” \textit{IEEE Trans. Wireless Commun.}, vol. 24, no. 4, pp. 2832-2846, 2025.



\bibitem{bb25.01.09.2}
F. Liu, Y. Zhang, Y. Xiong \textit{et al.}, “OFDM achieves the lowest ranging sidelobe under random ISAC signaling,” \textit{arXiv preprint	arXiv:2407.06691}, 2024.





\bibitem{bb25.01.09.3}
F. Liu,  Y. Xiong, S. Lu \textit{et al.}, “Uncovering the iceberg in the sea: Fundamentals of pulse shaping and modulation design for random ISAC signals,” \textit{IEEE Trans. Signal Process.}, early access, 2025.


\bibitem{bb24.03.15.33}
W. Yuan, L. Zhou, S. K. Dehkordi \emph{et al.}, “From OTFS to DD-ISAC: Integrating sensing and communications in the delay doppler domain," \textit{IEEE Wireless Commun.}, vol. 31, no. 6, pp. 152-160, 2024.


\bibitem{bb24.9.08.4}
J. Zhu, Y. Tang, F. Liu \textit{et al.}, “AFDM-based bistatic integrated sensing and communication in static scatterer environments," \textit{IEEE Wireless Commun. Lett.}, vol. 13, no. 8, pp. 2245-2249, 2024.


\bibitem{bb24.08.28.1}
F. Hlawatsch and G. Matz, \textit{Wireless Communications over Rapidly Time-Varying Channels}. Academic Press, 2011.




 
\end{thebibliography}
\end{document}